\documentclass[a4paper]{article}

\usepackage[margin=1in]{geometry}
\usepackage[pdftex,hyperindex,breaklinks]{hyperref}
\usepackage{amsmath,amsthm,bm}
\usepackage{algorithm}
\usepackage[noend]{algpseudocode}
\algnewcommand\Input{\item[\textbf{Input:}]}
\algnewcommand\Output{\item[\textbf{Output:}]}
\usepackage{xspace}
\usepackage[english]{babel}
\usepackage[utf8]{inputenc}
\usepackage[bitstream-charter]{mathdesign}

\newtheorem{thm}{Theorem}[section]
\newtheorem{lem}[thm]{Lemma}
\newtheorem{cor}[thm]{Corollary}
\newtheorem{prop}[thm]{Proposition}
\newtheorem{rem}[thm]{Remark}
\theoremstyle{definition}
\newtheorem{definition}{Definition}
\theoremstyle{remark}
\newtheorem{example}{Example}

\newcommand{\gO}{\mathcal{O}}
\newcommand{\gOt}{\tilde{\mathcal{O}}}
\newcommand{\gOeps}{\gO_\epsilon}
\newcommand{\cf}{\mathbb{F}}
\newcommand{\cfq}{{\mathbb{F}_q}}
\newcommand{\cfqs}{{\mathbb{F}_{q^s}}}
\newcommand{\gz}{\mathbb{Z}}
\newcommand{\gr}{\mathcal{R}}
\newcommand{\grext}{\mathcal{R}_{\text{ext}}}
\newcommand{\mul}{\mathsf{M}}
\newcommand{\imul}{\mathsf{I}}
\newcommand{\norm}[1]{\|#1\|_\infty}
\newcommand{\rp}{\textsc{RandomPrime}\xspace}
\newcommand{\coef}{\textsc{LeadingCoefficients}\xspace}
\newcommand{\scoef}{\textsc{SparseLeadingCoefficients}\xspace}
\newcommand{\verif}{\textsc{SparseVerification}\xspace}
\newcommand{\eval}{\textsc{ModularEvaluation}\xspace}
\newcommand{\seval}{\textsc{SparseModularEvaluation}\xspace}
\newcommand{\verify}{\textsc{ModularVerification}\xspace}
\newcommand{\verifyZ}{\textsc{ModularVerificationOverZ}\xspace}
\newcommand{\kaminski}{\textsc{KaminskiVerification}\xspace}
\DeclareMathOperator{\supp}{supp}
\newcommand\F[1][]{F^{\smash{[#1]}}}
\newcommand\f[1][]{f^{\smash{[#1]}}}
\newcommand\quo{\ \text{quo}\ }
\DeclareMathOperator\lcm{lcm}

\begin{document}
\title{Polynomial modular product verification and its implications}

\author{Pascal Giorgi \hfill  Bruno Grenet\hfill  Armelle Perret du Cray\\
    LIRMM, Univ. Montpellier, CNRS\\
    Montpellier, France\\
    \{pascal.giorgi,bruno.grenet,armelle.perret-du-cray\}@lirmm.fr}

\maketitle

\begin{abstract}
  Polynomial multiplication is known to have quasi-linear complexity in both the dense and the sparse cases. 
  Yet no truly linear algorithm has been given in any case for the problem, and it is not clear whether it is even possible. 
  This leaves room for a better algorithm for the simpler problem of verifying a polynomial product. 
  While finding
  deterministic methods seems out of reach, there exist probabilistic algorithms for the problem that are optimal in number of
  algebraic operations.  

  We study the generalization of the problem to the verification of a polynomial product modulo a sparse divisor. We investigate 
  its bit complexity for both dense and sparse multiplicands. In particular, we are able to show the primacy of the verification 
  over modular multiplication when the divisor has a constant sparsity and a second highest-degree monomial that is not too large.
  We use these results to obtain new bounds on the bit complexity of the standard polynomial multiplication verification. In 
  particular, we provide optimal algorithms in the bit complexity model in the dense case by improving a result of Kaminski and 
  develop the first quasi-optimal algorithm for verifying sparse polynomial product.
\end{abstract}

\section{Introduction}

Polynomials are one of the most basic objects in computer algebra and the study of fast polynomial operations remains a very
challenging task. Polynomials can be represented using either the dense representation, that stores all the coefficients in a
vector, or the more compact sparse representation, that only stores nonzero monomials. Depending on which representation is chosen
the problems might have a very different flavor leading to two very separate lines of research.

Polynomial multiplication is the most noticeable problem that attracted a lot of attention since many decades, culminating
nowadays with quasi-optimal algorithms \cite{cantor1991, harvey:polmul-nlogn}. Although such algorithms are really efficient in
theory and in practice, there are not yet optimal and they often rely on complex approaches that can be error prone. Therefore,
looking for rather simple procedure to verify the correctness of polynomial products is of great interest. From a theoretical
perspective, the goal is then to provide asymptotically faster algorithms than those for multiplying polynomials, ultimately
seeking for an optimal algorithm. In practice the objective is barely to find simpler and faster procedures that reveal easier to
trust.

In this work, we intend to present the most recent advances in verifying polynomial products in both the dense and sparse case, to
extend such results to either optimal algorithms or to more reliable solutions in practice. Finally, we extend the problem to some
specific modular multiplication of polynomials which seems to not having been explored yet.

\paragraph{Dense polynomial multiplication}
We know from the early 60's that dense polynomial arithmetic is subquadratic, and that it can even be quasi-linear when the
so-called FFT applies \cite{cooley1965}. It has been more than two decades later that Cantor and Kaltofen~\cite{cantor1991} provide a
quasi-linear algorithm without any assumption on the polynomial algebra. 
They show that two dense
polynomials of degree less than $n$ over an algebra $\mathcal A$ can be multiplied with $\gO(n \log n \log \log n)$ operations in $\mathcal A$.  In
regards to the bit complexity model, the operations in the base ring $\mathcal A$ cannot count $\gO(1)$ anymore, and the previous
algorithms may not lead to the best complexity estimates for specific domains such as $\mathcal A=\cfq$ or $\mathcal A =\gz$.  There, the use of
Kronecker substitution together with fast integer multiplication turns out to be the best alternative \cite[Section
8.4]{Gathen:2013}. It has been showed by Harvey and van~der~Hoeven in \cite{vdH:cyclomult} that one can reach a bit complexity of
$\gO(n \log q \log(n \log q)4^{\log^*(n)})$ for polynomial multiplication over $\cfq[X]$ for any prime field $\cfq$.  We shall
mention that very recently, such complexity have been further improved to $\gO(n \log q \log(n \log q))$ bit operations
\cite{harvey:polmul-nlogn} under some mild hypothesis. For polynomials with integer coefficients bounded by an integer $C$, the
complexity falls down to multiplying two integers of bit length $\gO(n(\log n+\log C))$ which gives
$\gO(n(\log^2n+\log C\log n+\log C \log \log C)=\gOt(n\log C)$\footnote{Here, and throughout the article, $\gOt(f(n))$ denotes
  $\gO(f(n)\log^k(f(n)))$ for some constant $k>0$.}  when we assume that $n$-bits integer multiplication complexity is
$\imul(n)=\gO(n \log n)$ \cite{harvey2019}.  For clarity in the presentation, we will often use $\mul(n)$ as the number of
operations in $\gr$ required to multiply two dense polynomials of size $n$, while $\mul_q(n)$ will denote the bit complexity for
such multiplication over a prime field $\cfq$.

\paragraph{Sparse polynomial multiplication} In the sparse representation, a polynomial $F = \sum_{i=0}^n f_i X^i \in\gr[X]$
is expressed as a list of pairs $(e_i, f_{e_i})$ such that all the $f_{e_i}$ are nonzero.  We denote by $\#F$ the \emph{sparsity}
of the polynomial $F$ which corresponds to its number of nonzero coefficients.  Let $F$ be a polynomial of degree $n$, and $\log C$ be a
bound on bit length of its coefficients.  Then, the size of the sparse representation of $F$ is $\gO(\#F(\log n+\log C))$ bits.
Contrary to the dense case, note that \emph{fast} algorithms for sparse polynomials must have a \mbox{(poly-)}logarithmic dependency on
the degree, and that the size of the output does not exclusively depend on the size of the inputs.  Indeed, the product of two
polynomials $F$ and $G$ has at most $\#F\#G$ nonzero coefficients. But it may have as few as $2$ nonzero coefficients, as shown by
the following example.

\begin{example}
  Let $F=X^{14}+2X^7 +2$, $G=3X^{13}+5X^8+3$ and $H=X^{14}- 2X^7 +2$. Then
  $FG= 3X^{27}+5X^{22}+6X^{20}+10X^{15}+3X^{14}+6X^{13}+10X^8+6X^7+6$ has nine terms, while $FH=X^{28}+4$ has only two.
\end{example}

Another difference with the dense case is that studying the complexity in a pure algebraic model remains meaningless, unless you
assume a transdichotomous model on the degree, meaning that the integer computation on the exponents is always $\gO(1)$
\cite{ColeHariharan, Nakos2020}.  The classical approach for computing the product of two polynomials of sparsity $T$ is to
generate all the $T^2$ possible monomials, and to sort them and merge those of equal degree to collect the monomials of the result. Using
radix sort, this algorithm takes for instance $\gO(T^2 (\imul(\log C) +\log n))$ bit operations over $\gz$ and it exhibits a $T^2$ factor in the space
complexity, whatever the number of terms in the result. 
Many improvements have been proposed to reduce this space complexity, to
extend the approach to multivariate polynomials, and to provide fast implementations in practice~\cite{Johnson74,
  MonaganPearce2009, MonaganPearce2011}.  Yet, none of these results reduces the $T^2$ factor in the time complexity. In general,
no complexity improvement is expected as the output polynomial may have as many as $T^2$ nonzero coefficients.  However, this
number of nonzero coefficients can be overestimated, giving the opportunity for output-sensitive algorithms.  Such algorithms have
first been proposed for special cases. Notably, when the output size is known to be small due to sufficiently structured
inputs~\cite{Roche2011}, especially in the multivariate case~\cite{vdHLec2012, vdHLebSch2013}, or when the support of the output
is known in advance~\cite{vdHLec2013}.

Output-sensitive multiplication algorithms try to take into account the two reasons that can decrease the sparsity of the product. The first
one is exponent collisions, while the second one occurs when these collisions imply some coefficient cancellations. The exponent
collision is captured by the \emph{sumset} of the exponents of $F = \sum_{i=1}^T f_i X^{\alpha_i}$ and
$G = \sum_{j=1}^T g_j X^{\beta_j}$, that is $\{\alpha_i+\beta_j:1\le i,j\le T\}$. Arnold and Roche call this set the
\emph{structural support} of the product $FG$ and its size the \emph{structural sparsity}~\cite{roche2015}. If $H = FG$, then the
structural sparsity $S$ of the product $FG$ satisfies $2\le \#H \le S\le T^2$. Observe that although $\#H$ and $S$ can be close,
their difference can reach $\gO(T^2)$ as shown by the next example.

\begin{example}\label{ex:structuralsupport}
  Let $F = \sum_{i=0}^{T-1} X^i$, $G = \sum_{i=0}^{T-1} (X^{iT+1}-X^{iT})$ and $H = FG$. We have $\#F = T$, $\#G = 2T$ and the structural
  sparsity of $FG$ is $T^2+1$ while $H = X^{T^2}-1$ has sparsity $2$.
\end{example}

For polynomials with non-negative integer coefficients, no coefficient cancellation can occur and Cole and Hariharan describe a
multiplication algorithm requiring $\gOt(S\log^2 n)$ operations in the RAM model with $\gO(\log(Cn))$ word
size~\cite{ColeHariharan}, where $\log(C)$ bounds the bitsize of the coefficients.  Arnold and Roche improve this complexity to
$\gOt(S\log n + \#H \log C)$ bit operations for polynomials with both positive and negative integer coefficients
~\cite{roche2015}. A recent algorithm of Nakos avoids the dependency on the structural sparsity for the case of integer
polynomials~\cite{Nakos2020}, using the same word RAM model as Cole and Hariharan.  Unfortunately, the bit complexity of this
algorithm, $\gOt((T\log +\#H\log^2 n) \log (Cn)+\log^3 n)$, is not quasi-linear.  More recently, we propose in
\cite{giorgi2020:sparsemul} the first quasi-optimal algorithm for sparse polynomial multiplication yielding a bit complexity of
$\gOt(T'(\log n+\log C))$ where $T'=\max(T,\#H)$. More precisely, taking $k=T'(\log n + \log c)$ which is the bit length of the
input and output, we are able to reach a bit complexity of $\gO(k \log^2 k \log^2 T (\log T+\log\log k))$.

\paragraph{Verification of polynomial products}  Considering the non-optimality of polynomial multiplication in both
representations, it is quite natural to ask whether it is rather a simple task or not to verify an instance of the problem. More
formally, given three polynomials $F,G$ and $H$, can we assert that $H$ is equal to the product of $F$ and $G$ in less operations
than computing the product itself? Furthermore, we want such procedure to be as simple as possible and to not rely on polynomial
multiplication if possible. Unfortunately, doing this with a deterministic procedure is not yet known, but using probabilistic
algorithms lead to positive answers as shown by several papers \cite{Demillo:1978,
  Schwartz:1980,Zippel:1979,giorgi2018}. 
Here and henceforth, polynomials are assumed to have coefficients in an integral domain $\gr$, rather than in a more general
algebra.

For dense polynomials this verification amounts to choosing a random element $\alpha$ in a finite subset of $\gr$ and to assert that
$H(\alpha) - F(\alpha)G(\alpha)$ is zero. In that case, the complexity for the verification becomes $\gO(n)$ operations in $\gr$,
which is optimal. Of course, the probability of error is less than one as soon as $\gr$ has more than $n$ elements.  If this not
the case, for instance when $\gr=\cf_2$, it is not desirable to choose $\alpha$ in a sufficient large extension of $\gr$ to have
$\gO(n)$ elements. The latter would require an extension of degree $\gO(\log n)$ and it would raise the complexity to
$\gO(n \mul(\log n))$. This is actually larger than the complexity $\mul(n)$ of computing the product. In \cite{Kaminski89}, Gamin's solved the latter problem by replacing the
evaluation, that corresponds to computing within $\gr[X]/(X-\alpha)$, by doing a polynomial multiplication within $\gr[X]/(X^i-1)$
for a random integer $i<n$. More precisely, by choosing $i = \gO(n^{1-e})$ for some $0<e<1/2$, his verification algorithm runs in $\gO(n)$
operations in $\gr$, whatever its size, with a probability of error bounded away from one. While the result sounds optimal from a
theoretical perspective, it might be mitigated for practical applications as it verifies polynomial multiplication of degree $n$
by doing multiplication of polynomials of degree $\gO(n^{1-e})$.

All these results remain valid under the bit complexity model, but the obtained complexity might not be optimal.  For polynomials
over $\cfq[X]$, both approaches using products in $\cfq[X]/(X-\alpha)$ or in $\cfq[X]/(X^i-1)$ lead to a bit complexity of
$\gO(n\imul(\log q))= \gO(n\log q \log \log q)$. 
While being non optimal, they remain however asymptotically faster than the
computation of the initial products by a factor $\log n/\log \log q$. 
Actually, Kaminski's approach has a better bit complexity than the standard method and can even yield a linear bit complexity in favorable cases.
For polynomials over $\gz[X]$, the result is more surprising
as it is possible to reach an optimal bit complexity of $\gO(n \log C)$ for any input. This result should be attributed to Kaminski, as he
provided in \cite{Kaminski89} all the necessary materials while not noticing the result explicitly. It seems surprising, but we
haven't found any references advertising such result. Thus, we propose to provide the description of those optimal verifications
of polynomial products.\medskip

For sparse polynomials, the verification of products remains less studied. It is misleading to think that using polynomial
evaluation is satisfactory. Assuming that only $T$ coefficients are nonzero, sparse polynomial evaluation is not quasi-linear in
the input size $\gO(T(\log n+\log C))$.  Indeed, computing $\alpha^n$ requires $\gO(\log n)$ operations in $\gr$ which implies a
complexity of $\gO(T\log n)$ operations in $\gr$ when applied to the $T$ nonzero monomials.  Since one needs to use a subset of
$\gr$ of size at least $n$ to ensure a nonzero probability of success, this implies that the bit complexity is at least $\gO(T\log^2
n)$. Using similar ideas as Kaminski's~\cite{Kaminski89}, we proposed recently in \cite{giorgi2020:sparsemul} to verify
sparse polynomial identities by doing the computation in $\gr[X]/(X^p-1)$ for a random prime $p$. In particular, we prove that
choosing $p=\gO(T^2 \log n)$ ensures that $(H-FG) \bmod (X^p-1) = 0$ implies that
$H-FG=0$ with good probability and that the computation can be done in quasi-linear time with $\gOt(T(\log n+\log C))$ bit operations.\medskip

Another important measure for randomized verification algorithms is the probability of failure. All the known verification algorithms are True-biased one-sided Monte Carlo algorithms. This means that they always return True if $H = FG$ and return False with probability at least $1-\epsilon$ otherwise. Given an algorithm with error probability at most $\epsilon$, we can attain any smaller probability of error $\tau$ by repeating $\gO(\frac{\log\tau}{\log\epsilon})$ rounds of the algorithm. This shows that the complexity of the algorithm is actually dependent on the target error probability. In our results, we always explicitly indicate this dependency. 
We can distinguish several \emph{regimes} of values for the error probability: the constant regime $\epsilon = \gO(1)$, the inverse polynomial regime $\epsilon = 1/n^{\gO(1)}$ and the inverse exponential regime $\epsilon = 1/2^{\gO(n)}$, where $n$ is the input degree. Given an algorithm with constant error probability, one can attain any smaller constant probability using a constant number of rounds. This keeps the same asymptotic complexity. The same is true for two probabilities inside the inverse polynomial regime. To get to the inverse polynomial regime from the constant regime, the number of rounds must be $\gO(\log n)$, slightly increasing the asymptotic complexity. The inverse exponential regime can then be attained using a polynomial number of rounds. In our context of linear and quasi-linear algorithms, the inverse exponential regime is not attainable in general. The best known verification algorithms have linear bit complexity in the inverse polynomial regime.

\paragraph{Contributions} As an extension of our prior work \cite{giorgi2020:sparsemul}, we propose to study more generally the
verification of polynomial multiplication in $\gr[X]/P$ where $P$ is a monic sparse polynomial. In the dense case, this
generalizes a work from one of the authors on the probabilistic verification on polynomial middle product~\cite{giorgi2018}.  By
reusing our modular product's verification, we show that we can address the difficulty of Kaminski's approach that verifies
polynomial products using products of roughly the same degree, more than $\sqrt n$. In particular, we show that we can avoid the
dependency on polynomial multiplication in every 
cases. 
When dealing with finite field arithmetic it is quite common
to rely on irreducible polynomials that are sparse \cite{MullenPanario:HandbookFiniteFields:2013}. Therefore, having the
possibility to verify multiplication over finite fields in less operations than computing the product seems of great interest. In
particular, we show that the verification of products in $\cfqs$ can be done in $\gO(s\#P)$ operations in $\cfq$ where
$P\in\cfq[X]$ is the monic irreducible polynomial of degree $s$ used to define $\cfqs$. Clearly for irreducible polynomial with
constant sparsity, as it is often the case over $\cf_2$~\cite{Golomb82} and more generally
$\cfq$~\cite{MullenPanario:HandbookFiniteFields:2013}, this offers an optimal verification procedure.  Finally, for sparse
polynomials, this work extends our prior result for $\gr[X]/(X^p-1)$ in \cite{giorgi2020:sparsemul} that was of great importance to
achieve the first quasi-linear time algorithm for sparse polynomial multiplication. We hope our new insight on this problem will
leverage other fast algorithms for sparse polynomial operations, especially for the division problem~\cite{Roche2018}. \medskip

All our techniques and results extend to the more general problem of verifying a polynomial identity of the form $\sum_i F_iG_i \bmod P = 0$, where the sum may have an arbitrary number of terms. It would be interesting to be able to extend these results to more general polynomial identities. As a very simple example, given $F_1$, $F_2$, $F_3$, $H$ and $P\in\gr[X]$ for some integral domain $\gr$, what is the complexity of the verification of $H = F_1F_2F_3\bmod P$? Obviously if the inputs are dense polynomials, the computation of $F_1F_2F_3\bmod P$ can be done in quasi-linear time. But the question is to design an algorithm than runs faster than performing the computation. In the sparse case, the computation may increase the input size quite a lot and even a quasi-linear time algorithm is lacking. More generally, the problem is to verify identities of the form $\sum_i\prod_j\sum_k\dotsb\prod_\ell F_{i,j,k,\dotsc,\ell} \bmod P = 0$. This problem can be phrased as a \emph{Modular Polynomial Identity Testing} (Modular PIT) problem. The standard Polynomial Identity Testing (PIT) problem takes as input an arithmetic circuit, or equivalently a straight-line program, and consists in deciding whether the polynomial it represents is zero. In this extension, a polynomial $P$ is also given as input and the question is whether the polynomial represented by the circuit is divisible by $P$. The standard PIT problem admits polynomial-time, and even quasi-linear-time, randomized algorithm. A very important open question is whether it also admits a polynomial-time \emph{deterministic} algorithm. For the Modular PIT problem, the question is already to design efficient \emph{randomized} algorithms. If the dense case, the challenge is to obtain faster algorithms than performing the product, ideally linear-time algorithms. In the sparse case, it is not even known how to solve the problem in randomized polynomial-time. Our results may be seen as a first step towards this goal.

\paragraph{Outline}
We start our work in Section~\ref{sec:prelim} by introducing all the technical materials that serve to demonstrate our main
results.  Then, Section~\ref{sec:eval} is devoted to the study of the evaluation of modular multiplication.  In particular, we
provide algorithms and their thorough analysis for evaluating $(FG)\bmod P$ on $\alpha$ without computing $FG\bmod P$.  The
results of that section serve to derive efficient algorithms in Section~\ref{sec:verif} for the verification of modular
multiplication of polynomials.  Finally, we present in Section~\ref{sec:product} the more general results on the verification of
classical polynomial multiplication. In particular, we extend the work of Kaminski~\cite{Kaminski89} for the dense case with
a thorough analysis of its bit complexity that enables to reach optimal verification. We also give a more detailed presentation
of our first quasi-optimal algorithm for the sparse case that appears in \cite{giorgi2020:sparsemul}.

\section{Preliminaries}\label{sec:prelim}

\subsection{Notations and complexities}

Let $Q\in\gr[X]$ be a degree-$n$ polynomial. We denote its coefficient of degree $i$ by $q_i$. The \emph{sparsity} of $Q$ is its
number of nonzero monomials and is denoted by $\#Q$. The \emph{support} of $Q$ is the set $\supp(Q) = \{i : q_i\neq 0\}$. If $Q$
is a polynomial over $\gz$, we denote by $\norm{Q}$ its norm, defined as $\max_{0\le i\le n} |q_i|$.
We denote by $\log(\cdot)$ the base-$2$ logarithm and by $\ln(\cdot)$ the natural logarithm. We also use $\log_b(\cdot)$ to denote the base-$b$ logarithm defined by $\log_b(x) = \frac{\log x}{\log b}$.

We work in this paper with dense and sparse polynomials. A dense polynomial is represented as the vector of its coefficients,
which has size $n+1$ for a degree-$n$ polynomial. A sparse polynomial is represented by the list of its nonzero monomials. We
consider that we work, for sparse polynomials, with an abstract structure of \emph{sparse vector}. In practice, this can be
implemented by several data structures, depending on the operations that need to be performed. A standard choice in sparse
polynomial arithmetic is the use of heaps~\cite{Johnson74,MonaganPearce2009,MonaganPearce2011}. To get better complexities, we
might resorting to van Emde Boas Trees~\cite[Chapter 20]{CormenBook} as in Section~\ref{ssec:sparseval}.We also use sparse vectors in some
algorithms to represent data which are not directly polynomials. The underlying data structure is the same as for sparse
polynomials.

\paragraph{Complexity of dense polynomial multiplications}
We denote by $\mul(n)$ the number of ring operations needed to compute a product of degree-$n$ dense polynomials over an integral
domain. We can take $\mul(n)=\gO(n\log n\log\log n)$~\cite{cantor1991}. We denote by $\mul_{q}(n)$ the bit complexity of the
multiplication of two degree-$n$ dense polynomials over a finite field $\cfq$. The best known bounds on $\mul_q(n)$ are
$\gO(n\log q\log(n\log q)\,4^{\log^* (n\log q)})$~\cite{vdH:cyclomult} unconditionally and $\gO(n\log q\log(n\log q))$ assuming
the existence of some Linnik constant~\cite{harvey:polmul-nlogn}.  To simplify the notation, we assume the existence of this
Linnik constant.  The cost of multiplying two elements in an extension field $\cf_{q^d}$ is the cost of degree-$d$ polynomial
multiplications, that is $\gO(\mul_q(d))$.  Let $F$, $G\in \gz[X]$ of degree $n$ and norm $C$. Their product has norm at most
$nC^2$. To compute $FG$, we can evaluate both $F$ and $G$ on some power of $2$ larger than $nC^2$, multiply the resulting integers
(that have size $n\log(nC^2)$), and read the coefficients of $FG$ directly on the output integer. Let $\imul(m) = \gO(m\log m)$
denote the bit complexity of multiplying two $m$-bit integers~\cite{harvey2019}. Then the bit complexity of multiplying $F$ and
$G$ is $\imul(n\log(nC^2)) = \gO(n\log^2n+n\log n\log C + n\log C\log\log C)$.
 
\paragraph{Complexity of polynomial evaluation}
The evaluation of a dense degree-$n$ polynomial $F\in\gr[X]$ on a point $\alpha\in\gr$ requires $\gO(n)$ operations in $\gr$ using
for instance Horner scheme. If $\alpha$ lies in an extension ring $\grext$ of $\gr$, the evaluation requires $\gO(n)$ operations
in $\grext$. If $F$ has coefficients in a finite field $\cfq$, this translates directly to a linear number of operations in
$\cfq$. Now if $F$ has coefficients in $\gz$, one must take into account the growth of the integers during the computation. Using
a divide-and-conquer approach to use balanced integer multiplications, the cost of the evaluation is
$\gO(\imul(n\log C)\log(n\log C))$ bit operations where $C = \max(|\alpha|,\norm{F})$.  We note that this cost is quasi-linear in
the worst case output size while using Horner scheme would have been quadratic.

To evaluate a sparse polynomial $F\in\gr[X]$ on $\alpha\in\gr$, we compute the relevant powers of $\alpha$ and then perform $\#F$
multiplications and additions in $\gr$. Computing each power independently yields $\gO(\#F\log n)$ operations in
$\gr$. Using simultaneous exponentiation~\cite{Yao1976}, the cost is reduced to $\gO(\log n+\#F\log n/\log\log n)$ operations in
$\gr$. Again, this directly translates to operations in $\cfq$ if $\gr=\cfq$. For polynomials with integer coefficients, the
growth is much more severe than in the dense case. Indeed, $\alpha^n$ has $\gO(n\log|\alpha|)$ bits. This implies that the bit
complexity is at least linear in $n$ which is exponentially larger than the input size. The cost is actually not better than with
dense polynomials.

\subsection{Bounds on polynomial products and modular reductions}

Reducing a polynomial modulo $P$ changes its norm and sparsity. We provide bounds on these growths. They rely on the \emph{gap} between the degree of $P$ and its \emph{second degree}, that is the degree of its second highest-degree monomial. 

\begin{definition} Let $P = X^n + \sum_{i=0}^k p_i X^i$ for $k < n$ and $p_k\neq 0$. The \emph{second degree} of $P$ is the integer $k$. The \emph{gap parameter} $\gamma$ of $P$ is $\gamma = \frac{1}{n}(n-k)$. 
\end{definition}

In particular, the second degree of $P$ is $(1-\gamma)n$. The parameter $\gamma$ is between $0$ and $1$. If $\gamma$ is close to
$0$, the polynomial actually has no gap, while $\gamma = 1$ corresponds to a binomial $aX^n+b$. We note that given this
definition, $\frac{1}{\gamma}$ is always upper bounded by $n$. Polynomials with a large gap are also known as \emph{sedimentary}
polynomials~\cite{MullenPanario:HandbookFiniteFields:2013}. A polynomial is said $t$-sedimentary if it is of the form $X^n + H$
where $\deg(H) = t$. A $t$-sedimentary polynomial is a polynomial with gap parameter $(n-t)/n$ and conversely a monic polynomial
with a gap parameter $\gamma$ is $(1-\gamma)n$-sedimentary.

The norm of the product of two polynomials is classically related to their norms and degrees. This can be slightly refined using the
sparsities instead of the degrees.

\begin{lem}\label{lem:prodbounds}
Let $F$ and $G$ be two polynomials over $\gz$. Then
$\norm{FG} \le \min(\#F, \#G)\norm F\norm G$. 
\end{lem}

\begin{proof}
Let $H = \sum_k h_k X^k$ be the product of $F$ and $G$. Then $h_k = \sum_{i+j = k} f_i g_j$. Let $T = \min(\#F,\#G)$. Then the sum to define $h_k$ has size of most $T$. Since $|f_i|\le\norm F$ for all $i$ and $|g_j|\le\norm G$ for all $j$ by definition, $|h_k|\le T\norm F\norm G$, whence the result.
\end{proof}

The modular reduction of polynomials has a bigger impact on the norm. It is actually related to several parameters such as the gap parameter of the divisor and the difference of the degrees. The following example shows a large increase in the norm, as well as a densification of the result.

\begin{example}
	Let $P=X^{80}+7X^{65}+2X^{61}-8X^{59}+X^{56}+3$ and $Q=X^{131}+4X^{120}+8X^{118}-3X^{108}-3X^{80}+X^{71}+5X^{32}$.
	Here the gap parameter of $P$ is $\gamma=\frac{3}{16}$, $\#P=6$, $\#Q = 7$ and $\norm{P}=\norm{Q}=8$.
	The polynomial $Q\bmod P$ has degree $79$, sparsity $53$ and norm $11912$. 
\end{example}

The following proposition bounds the growth on the different parameters of the polynomial after a modular reduction.

\begin{prop}\label{prop:sparsity+coeffbound}
	Let $Q$ be a sparse polynomial of degree at most $n-1+k$ and $P$ a monic polynomial of degree $n$ with $\#P\geq 2$.
	The polynomial $Q\bmod P$ has at most $\#Q(\#P-1)^{\lceil\frac{k}{\gamma n}\rceil}$ monomials.
	If $Q$ and $P$ are defined over $\gz$, 
        $\norm{Q\bmod P}\le \norm{Q}(\#P\norm{P})^{\lceil\frac{k}{\gamma n}\rceil}$.
\end{prop}

\begin{proof}
	
  We analyse the growth of the norm and the sparsity while performing the euclidean division. 
	
        Instead of following the classical quadratic algorithm, 
	we first reduce once all the monomials of $Q$ with degree at least $n$ to obtain a new dividend. We repeat this process until the dividend has degree less than $n$. 
	Let us define the sequence $(Q^{[i]})_i$ by $Q^{[0]}=Q$ and 
	$Q^{[i+1]} = (Q^{[i]}\bmod X^n) + (Q^{[i]}\quo X^n)(X^n-P)$. Then $Q^{[i]}\bmod P = Q\bmod P$ for all $i$. 
        Since $\deg(Q^{[i]}\quo X^n) = \deg(Q^{[i]})-n$ and $\deg(X^n-P)\le (1-\gamma)n$, $\deg(Q^{[i+1]})\le\max(n-1,\deg(Q^{[i]})-\gamma n)$, whence $\deg(Q^{[i]})\le\max(n-1, \deg(Q)-i\gamma n)$.

        Also, $\#Q^{[i+1]}$ is at most $\#Q^{[i]}(\#P-1)$, thus $\#Q^{[i]}\le \#Q(\#P-1)^i$. Finally, 
        \[\norm{Q^{[i+1]}} \le \norm{Q^{[i]}}(1+\min(\#Q^{[i]},\#P-1)\norm P).\]
        Therefore, $\norm{Q^{[i]}}\le (\#P\norm P)^i \norm Q$.
	
        Since $\deg(Q^{[i]}) \le n+k-1-i\gamma n$, $\deg(Q^{[i]}) < n$ if $i=\lceil\frac{k}{\gamma n}\rceil$. This implies that $Q^{[i]} = Q\bmod P$. 
\end{proof}

\subsection{Random primes and random irreducible polynomials}

We collect in this section some useful results to produce random prime numbers and random irreducible polynomials over finite fields.

\begin{prop}[\cite{rosser1962}] \label{prop:rosser}
If $\lambda\ge21$, there are at least $\frac{3}{5}\lambda/\ln\lambda$ prime numbers in $[\lambda,2\lambda]$.
\end{prop}

Using this proposition together with Miller-Rabin probability test, we can produce integers that are prime with good probability~\cite{shoup2008}.

\begin{prop}\label{prop:rdprime}
There exists an algorithm $\rp(\lambda,\epsilon)$ that returns an integer $q$ in $[\lambda, 2\lambda]$, such that $q$ is prime with probability at least $1-\epsilon$. 
	It requires $\gO(\log(\frac{1}{\epsilon})\log^2(\lambda)\imul(\log\lambda)\log\log\lambda)$ bit operations.
\end{prop}

\begin{prop}\label{prop:choixdep}
  Let $H\in\gr[X]$ be a nonzero polynomial of degree at most $n$ and sparsity at most $T$, $0<\epsilon<1$ and
  $\lambda = \max(21, \frac{10}{3\epsilon} T\ln n)$. With probability at least $1-\epsilon$,
  $\rp(\lambda, \frac{\epsilon}{2})$ returns a prime number $p$ such that $H\bmod X^p-1\neq 0$.
\end{prop}

\begin{proof}
  It is sufficient, for $H\bmod X^p-1$ to be nonzero, that there exist one exponent $e$ of $H$ that is not congruent to any other
  exponents $e_j$ modulo $p$.  In other words, it is sufficient that $p$ does not divide any of the $T-1$ differences
  $\delta_j=e_j-e$.
  
  Noting that $\delta_j\leq n$, the number of primes in $[\lambda,2\lambda]$ that divide at least one $\delta_j$ is at most
  $\frac{(T-1)\ln n}{\ln \lambda}$. Since there exists $\frac{3}{5}\lambda/\ln \lambda$ primes in this interval, the probability
  that a prime randomly chosen from it divides at least one $\delta_j$ is at most $\epsilon/2$. $\rp(\lambda,\epsilon/2)$ returns
  a prime in $[\lambda,2\lambda]$ with probability at least $1-\epsilon/2$, whence the result.
\end{proof}

The following two propositions will be useful to either reduce integer coefficients modulo some prime number or to construct irreducible polynomials over finite fields.

\begin{prop}\label{prop:choixdeq}
  Let $H\in\gz[X]$ be a nonzero polynomial, $0<\epsilon<1$ and $\lambda \geq \max(21,\frac{10}{3\epsilon}\ln \|H\|_\infty)$.  Then
  with probability at least $1-\epsilon$, $\rp(\lambda, \frac{\epsilon}{2})$ returns a prime $q$ such that
  $H\bmod q \neq 0$.
\end{prop}

\begin{proof}
  Let $h_i$ be a nonzero coefficient of $H$,
  a random prime from $[\lambda,2\lambda]$ divides $h_i$ with probability at most
  $\frac{5}{3} \ln \|H\|_\infty/\lambda\leq \epsilon/2$. Since $\rp(\lambda,\epsilon/2)$ returns a prime in $[\lambda,2\lambda]$
  with probability at least $1-\epsilon/2$ the result follows.
\end{proof}

\begin{prop}[{\cite[Chapter 19]{shoup2008}}]\label{prop:nbirrpoly}
 The number of irreducible monic polynomial of degree $d$ over a field $\cfq$ is between $\frac{q^d}{2d}$ and $\frac{q^d}{d}$.
\end{prop}

\begin{prop}[{\cite[Chapter 20]{shoup2008}}]\label{prop:irrpoly}
  There exists an algorithm that, given a finite field $\cfq$, an integer $d$ and $0<\epsilon<1$, computes a degree-$d$ polynomial
	in $\cfq[X]$ that is irreducible with probability at least $1-\epsilon$. It requires $\gO(\log(\frac{1}{\epsilon})d^2\mul(d)(\log q+\log\log d))$ operations in $\cfq$
	or $\gO(\log(\frac{1}{\epsilon})d^4\log q)$ operations in $\cfq$ if using only naive polynomial multiplications.
\end{prop}

\begin{rem}
 Shoup~\cite{shoup2008} presents Las Vegas algorithms for Propositions~\ref{prop:rdprime} and~\ref{prop:irrpoly}. We consider Monte Carlo versions of his algorithms. Also, he analyses the complexities with naive algorithms. Our complexity estimates use fast integer and polynomial arithmetic.
\end{rem}

\section{Evaluation for polynomial multiplication in a quotient ring}\label{sec:eval}

As seen earlier, the verification of polynomial multiplication mainly relies on the evaluation of the polynomial identity at a
random point.  In this section we present algorithms to efficiently compute the evaluation of a modular product $(FG)\bmod P$ on a
point $\alpha$, without computing $(FG)\bmod P$. There, the modulus $P$ is always considered as a sparse polynomial, while $F$ and
$G$ can be either dense or sparse. 

Section~\ref{ssec:evalmodbinom} describes our method in the simpler case where $P$ is a binomial. We obtain linear-time
evaluations, whether $F$ and $G$ are dense or sparse.
Section~\ref{ssec:eval} generalizes the method to the product of two dense polynomials \emph{modulo} a sparse modulus, and Section~\ref{ssec:sparseval} presents the case of a sparse modular product.

\subsection{Evaluation of a product modulo a binomial}\label{ssec:evalmodbinom}
Let us first present our method to evaluate a modular product $FG\bmod P$ where $P = X^n-1$.
This special case illustrates our more general method. It also has its own interest since it is used as the main tool for the
verification of a product of two polynomials in Section~\ref{sec:product}, either for dense or sparse representation.

We first describe the algorithm for dense polynomials $F$ and $G$.

\begin{thm} \label{lem:evalmodbinomial} Let $F$ and $G$ be two polynomials in $\gr[X]$ of degrees less than $n$ and
  $\alpha\in\gr$. The polynomial $(FG)\bmod X^n-1$ can be evaluates on $\alpha$ using $\gO(n)$ operations in $\gr$.
\end{thm}

\begin{proof}
  Let $H = FG$ and $M = H\bmod X^n-1$. We denote by $f_i$ (resp. $g_i$, $h_i$, $m_i$) the coefficient of degree $i$ of the
  polynomial $F$ (resp. $G$, $H$, $M$). Let also $\vec g = (g_0,\dotsc,g_{n-1})^T$, $\vec h = (h_0,\dotsc,h_{2n-2})^T$ and
  $\vec m = (m_0,\dotsc, m_{n-1})^T$.

  It is a well-known fact that considering $F$ as fixed, the multiplication by $F$ is a linear map described by a Toeplitz
  matrix. More precisely, we have $\vec h = T_F \vec g$ where
\[T_F = \begin{pmatrix}
    f_0 \\
    f_1 & f_0 \\
    \vdots&&\ddots\\
    f_{n-1} & \dots &\dots& f_0 \\
     & f_{n-1} &  & f_1 \\
     &&\ddots&\vdots\\
     &  & & f_{n-1}
    \end{pmatrix}.
\]
Since $M = H\bmod X^n-1$, $m_i = h_i+h_{i+n-1}$ for $0\le i < n-1$ and $m_{n-1} = h_{n-1}$. Therefore, $\vec m = C_F\vec g$ where $C_F$ is the circulant matrix
\[C_F = \begin{pmatrix}
    f_0 & f_{n-1} & \cdots & f_1\\
    f_1 & f_0     & \cdots & f_2\\
    \vdots & \vdots&&\vdots\\
    f_{n-1} & f_{n-2} & \cdots & f_0
    \end{pmatrix}.
\]

On the other hand, evaluating $M$ on $\alpha$ corresponds to the inner product $\vec\alpha_n\vec m$ where
$\vec\alpha_n = (1,\alpha,\dotsc,\alpha^{n-1})$. Therefore, our aim is to compute $\vec\alpha_n C_F\vec g$. The standard way to
perform this evaluation corresponds to first computing $\vec m = C_F\vec g$ and then $\vec\alpha_n\vec m$. As noticed by
Giorgi~\cite{giorgi2018}, the bracketing $(\vec\alpha_n C_F)\vec g$ yields a faster algorithm due to the structure of the matrix
$C_F$.

Let $\vec c = \vec\alpha_n C_F$. Then
$c_{j+1} = \sum_{\ell=0}^{n-1} \alpha^\ell f_{(\ell-j-1)\bmod n} = f_{n-j-1} + \alpha\sum_{\ell=0}^{n-2} \alpha^\ell f_{(\ell-j)\bmod
  n}$. Since for $j>0$ we have $\sum_{\ell=0}^{n-2} \alpha^\ell f_{(\ell-j)\bmod n} = c_{j}-\alpha^{n-1}f_{n-j-1}$, we obtain
the recurrence relation

\begin{equation}\label{eq:recurrence}
  \begin{cases}
    c_{j+1} = \alpha c_{j}-P(\alpha)f_{n-j-1} & \text{\, for } j\geq 0\\
    c_0 = F(\alpha) &
  \end{cases}
\end{equation} 
where $P = X^n-1$ and $c_0=F(\alpha)$.

It is immediate that exploiting such recurrence relation for computing the evaluation of $(FG) \bmod P$ leads to a complexity of
$\gO(n)$ operations in $\gr$. Indeed, once $c_0$ and $P(\alpha)=\alpha^n-1$ are computed each other $c_j$ can be computed sequentially at cost $\gO(1)$.
\end{proof}

For completeness, we provide the full description of this method in Algorithm~\ref{alg:evalbinom}.

\begin{algorithm}\label{alg:evalbinom}
\caption{\textsc{EvaluationModuloBinomial}}
\begin{algorithmic}[1]
  \Input $F$, $G\in\gr[X]$ with $\deg(F),\deg(G) < n$, and $\alpha\in\gr$.
 
  \Output $(FG\bmod X^n-1)(\alpha)$

\State $c\gets F(\alpha)$
\State $P_\alpha\gets\alpha^n-1$
\State $\beta \gets c g_0$
\For{$j=1$ to $n-1$} 
    \State $c\gets\alpha c-P_\alpha f_{n-j}$
    \State $\beta\gets\beta+cg_j$
\EndFor
\State \Return $\beta$
\end{algorithmic}
\end{algorithm}

We can actually be more precise on the number of operations required by Algorithm~\ref{alg:evalbinom}. In particular when $\alpha$
does not lie into $\gr$ but in an extension $\grext$ of $\gr$, we can distinguish between operations in $\gr$ and $\grext$. In the
next corollary, we call \emph{scalar multiplications} those that are multiplications of an element of $\grext$ by an element of
$\gr$. The following analysis minimizes the number of non-scalar multiplications.

\begin{cor}
  Let $F$ and $G$ be two polynomials in $\gr[X]$ of degree less than $n$ and $\alpha\in\grext$. The polynomial $(FG)\bmod X^n-1$
  can be evaluated on $\alpha$ using $2n-2$ multiplications and $3n-2$ additions in $\grext$, and $3n-2$ scalar multiplications.
\end{cor}

\begin{proof}
	We can first compute $\alpha^2$, $\alpha^3$, \dots, $\alpha^n$ using $(n-1)$ multiplications. Then, $F(\alpha)$ can be computed using $(n-1)$ scalar multiplications and additions, and $P(\alpha) = \alpha^n-1$ require one more addition. The initial value $cg_0$ of $\beta$ requires one scalar multiplication. Then each iteration of the loop require one multiplication, two scalar multiplications and two additions. Therefore, the complete evaluation require $3n-2$ additions, $2n-2$ multiplications and $3n-2$ scalar multiplications.
\end{proof}

To minimize the total number of multiplications instead, we remark that one can also evaluate $F$ using Horner's scheme with $(n-2)$ multiplications, one scalar multiplication and $(n-1)$ additions. Then $\alpha^n$ has to be computed using at most $2\log n$ multiplications. This results in $3n-2$ additions, $2n-3+2\log n$ multiplications and $2n$ scalar multiplications. The total number of multiplications (scalar or not) is a bit less.

We now turn to the analysis of the algorithm for $F$ and $G$ given in sparse representation.

\begin{thm}\label{thm:sparseevalbin}
	Let $F$ and $G$ be two sparsely represented polynomials in $\gr[X]$ of degrees less than $n$ and $\alpha\in\gr$. The polynomial $(FG)\bmod X^n-1$ can be evaluated on $\alpha$ using $\gO((\#F+\#G)\log n)$ operations in $\gr$. 
\end{thm}

\begin{proof}
We use the same notations as in the previous proof. If the support of $G$ is 
$\supp(G) = \{j_0,\dotsc, j_{\#G-1}\}$ with $j_0 < \dotsb < j_{\#G-1}$, the inner product $\vec c\vec g$ is equal to $\sum_{k=0}^{\#G-1} c_{j_k} g_{j_k}$. This means that only the $\#G$ entries $c_{j_0}$, \dots, $c_{j_{\#G-1}}$ of $\vec c$ need to be computed.
Applying the recurrence relation~\eqref{eq:recurrence} as many times as necessary, we obtain the new recurrence relation 
\begin{equation}\label{eq:sparserec}
  \begin{cases}
    c_{j_{k+1}}=\alpha^{j_{k+1}-j_k}c_{j_k} -P(\alpha)\sum_{\ell=j_k+1}^{j_{k+1}}\alpha^\ell f_{n-\ell}& \text{ for } k\geq 0 \\
    c_{j_0}= ((X^{j_0}F)\bmod X^n-1) (\alpha).
  \end{cases}
\end{equation}
The initial value $c_{j_0}$ can be computed using $\gO(\#F\log n)$ operations in $\gr$ since it needs $\#F$ exponentiations of
$\alpha$ with exponent bounded by $n$. Most values of $f_{n-\ell}$ are actually equal to zero since $F$ is sparse.

A nonzero coefficient $f_t$ of $F$ appears in the definition of $c_{j_{k+1}}$ if and only if $n-j_{k+1}\leq t< n-j_{k}$.  Thus,
each $f_t$ is used exactly once to compute all the $c_{j_k}$'s.  Since for each summand, one needs to compute $\alpha^\ell$ for
some $\ell<n$, the total cost for computing all the sums is $\gO(\#F\log n)$ operations in $\gr$. Similarly, the computation of
	$\alpha^{j_{k+1}-j_k}c_{j_k}$ for all $k\in[0,\#G-2]$ costs $\gO(\#G\log n)$ operations in $\gr$
	plus $\#G-1$ additions of $\gO(\log n)$-bit integers to get the exponents.
	As one operation in $\gr$ requires at least one bit operation, the integer additions that costs $\gO(\#G\log n)$ bit operations are negligible.
	The last remaining step is the final inner product which costs
$\gO(\#G)$ operations in $\gr$, whence the result. 
\end{proof}

As in the dense case, one can be more precise on the complexity if $\alpha$ liesin an extension $\grext$.
In contrary to the dense case where there is more operations in $\gr$ than in $\grext$, one can note that the number of operations in $\gr$ is negligible in the sparse case.

\begin{cor} \label{cor:sevalbin} Let $F$ and $G$ be two sparsely represented polynomials in $\gr[X]$ of degrees less than $n$ and
  $\alpha\in\grext$. The polynomial $(FG)\bmod X^n-1$ can be evaluated on $\alpha$ using
	$\log n + \gO((\#F+\#G)\log n/\log\log n)$ operations in $\grext$ plus $\#G-1$ additions of $\gO(\log n)$-bit integers.
\end{cor}

\begin{proof}
	We first notice that in the sparse case the operations on $\alpha$ dominate the complexity. These operations are operations in $\grext$.
  To improve the complexity estimates, we remark that in the sparse settings we need to compute $\alpha^t$ for several values of
  $t$. The computation of $c_{j_0}$ requires to know $\#F$ values of $\alpha^t$, more precisely those with $t = \ell-j_0\bmod n$
  for each nonzero coefficient $f_\ell$ of $F$. To apply Equation~\eqref{eq:sparserec}, one needs to compute $\alpha^t$ for
  $t = j_{k+1}-j_k$, $1\le k<\#G$, and for $t = \ell$ where $f_{n-\ell}\neq 0$. The value $\alpha^n$ is also needed to compute
  $P(\alpha)$. Finally, the inner product requires to compute $\alpha^t$ for each nonzero $g_t$. Altogether, one needs $\alpha^t$
  for at most $2(\#F+\#G)$ values of $t$, each at most $n$. They can be computed independently using fast exponentiation, using at
  most $\gO((\#F+\#G)\log n)$ multiplications, as it is done in Theorem~\ref{thm:sparseevalbin}. Actually, Yao~\cite{Yao1976}
  shows that these values of $\alpha^t$ can be computed simultaneously using only $\log n + \gO((\#F+\#G)\log n/\log\log n)$
  multiplications.  Once these $\alpha^t$ have been computed, computing $c_{j_0}$ and the $c_{j_k}$'s by means of
  Equation~\eqref{eq:sparserec}, as well as the inner product $\vec c\vec g$, only require $\gO(\#F+\#G)$ operations.

\end{proof}

\subsection{Evaluation of a dense modular product}\label{ssec:eval}

In this section, we extend the previous algorithm to the evaluation of a polynomial $FG\bmod P$ where $P$ is any monic sparse
polynomial. We first consider the case where $F$ and $G$ are given in dense representation. The case where they are given is
sparse representation is postponed to the next section.

The algorithm goes along the same lines as the evaluation modulo $X^n-1$.
Let $\F[i] = (X^iF)\bmod P$. We can rewrite $FG\bmod P=\sum_{i=0}^{n-1}g_i\F[i]$ where $g_i$ is the coefficient of degree $i$ in
$G$. The evaluation of this equality on a point $\alpha$ yields the
formula 
\begin{equation}\label{eq:eval}
(FG\bmod P)(\alpha)= \sum_{i=0}^{n-1}g_i\F[i](\alpha).
\end{equation}
To make use of this formula, we need to be able to efficiently evaluate each $\F[i]$ on $\alpha$. Note that consecutive $\F[i]$'s
are bound by the recurrence relation $F^{[i+1]}=(X\F[i])\bmod P$. Since $\deg(\F[i]) = n-1$,
$(X \F[i])\bmod P = X\F[i] - \f[i]_{n-1}P$ where $\f[i]_{n-1}$ is the coefficient of degree $n-1$ of $\F[i]$.  Consequently we
have the following recurrence relation
\begin{equation}\label{eq:rec}
  \begin{cases}
    F^{[i+1]}(\alpha)=\alpha \F[i](\alpha)-\f[i]_{n-1}P(\alpha) & \text{ for } i\geq 0 \\
    F^{[0]}(\alpha) = F(\alpha)
  \end{cases}
\end{equation}
The evaluations of each $\F[i]$ on $\alpha$ can thus be computed iteratively from $F(\alpha)$, only knowing the coefficient
$\f[i]_{n-1}$ of $\F[i]$ for $0 < i < n-1$.

We first present an algorithm to compute these coefficients. Note that we did not need such an algorithm when $P = X^n-1$ since
these coefficients were given for free as we had $\f[i]_{n-1} = f_{n-1-i}$. In the general case, the computation is based on the
recurrence relation $\F[i+1] = X\F[i]-\f[i]_{n-1}P$, which implies
\begin{equation}\label{eq:recFi}
\f[i+1]_{k} = \f[i]_{k-1} - \f[i]_{n-1}p_k
\end{equation}
for $0 < k \le n-1$. This allows to compute each $\f[i]_{n-1}$, starting from the values of $\f[0]_{k}$ for all $k$. These values are given as input since $\F[0] = F$ by definition. Note that since $P$ is a sparse polynomial, Equation~\eqref{eq:recFi} actually reduces to an equality $\f[i+1]_k = \f[i]_{k-1}$ in many cases. Algorithm~\ref{algo:coef} takes this into account and only performs the required updates.

\begin{algorithm}
\caption{\coef}
\label{algo:coef}
\begin{algorithmic}[1]
  \Input Two polynomials $P$ and $F$ in $\gr[X]$, with $\deg(F) < \deg(P)=n$ and $P$ monic. 
 
  \Output The vector $[\f[0]_{n-1}, \f[1]_{n-1}, \dotsc, f^{[n-2]}_{n-1}]$, where $\f[i]_{n-1}$ is the coefficient of degree $n-1$ of $\F[i] = (X^iF)\bmod P$. 

\State $V\gets [f_{n-1}, f_{n-2}, \dotsc, f_1]$ 
    \label{coef:init}
\For{$i=0$ to $n-2$} 
	\For{$k\in\supp(P)$ such that $i < k < n$} 
		\State $V[i+n-k] \gets V[i+n-k] - p_kV[i]$ \label{coef:update}
	\EndFor
\EndFor
\State \Return $V$
\end{algorithmic}
\end{algorithm}

\begin{lem} \label{lem:lastcoefdense}
Algorithm~\ref{algo:coef} is correct. It uses $\gO(n\#P)$ operations in $\gr$. 
\end{lem}

\begin{proof}
  The number of operations is clear: all operations are performed at Step~\ref{coef:update} and it is called $\gO(n\#P)$
  times. Note that the external for loop can be stopped as soon as there exists no $k\in\supp(P)$ such that $i < k < n$. In other
  words, $i$ never goes beyond
  $\deg(X^n-P)-1$. 

  To show the correctness of Algorithm~\ref{algo:coef}, we prove by induction that after iteration $i$ of the external loop, the following property
  $\mathcal{P}(i)$ holds:
  \[
    V[j] = \f[j]_{n-1}  \text{ for any } j\le i+1 \text{ and }
    V[j] = f^{[i+1]}_{n-(j-i)} \text{ for } j > i+1.
  \]

Before the first iteration, $\mathcal{P}(-1)$ holds since it reads $V[j] = \f[0]_{n-j-1}$ for all $j$, and $F = \F[0]$ by definition.

Suppose that $\mathcal{P}(i-1)$ holds. In particular, $V[j] = \f[j]_{n-1}$ for $j\le i$. During iteration $i$, only $V[i+1]$ to
$V[n-2]$ can be modified so these equalities remain after that iteration.
For $j > i$, $V[j] = \f[i]_{n-(j-i+1)}$ before the iteration by hypothesis. After the iteration, it becomes
\[
  V[j]=\f[i]_{n-(j-i+1)}-p_{n-j+i}V[i] = \f[i]_{n-j+i-1}-p_{n-j+i}\f[i]_{n-1}.
\]
Equation~\eqref{eq:recFi} shows that $V[j] = f^{[i+1]}_{n-j+i}$ after Step~\ref{coef:update}, and $\mathcal{P}(i)$ holds.

To conclude, after the last iteration, $V[j] = \f[j]_{n-1}$ for all $j\le n-2$ and the algorithm is correct.
\end{proof}

We can now make use of Algorithm~\ref{algo:coef} to evaluate $FG\bmod P$ on a point $\alpha$. In the following algorithm, we assume that $\alpha$
belongs to some extension ring $\grext$ of $\gr$. Our analysis distinguishes between operations in $\gr$ and in $\grext$.  
\begin{algorithm}
\caption{\eval}
\label{algo:eval}
\begin{algorithmic}[1]
  \Input $P$, $F$, $G\in\gr[X]$ with $\deg(F),\deg(G) < \deg(P)=n$, $P$ monic, and $\alpha\in\grext$.
 
  \Output $(FG\bmod P)(\alpha)$

\State $V\gets [\f[0]_{n-1},\dotsc,f^{[n-2]}_{n-1}]$ using a call to $\coef(P, F)$
\State $P_\alpha\gets P(\alpha)$ 
\State $F_\alpha\gets F(\alpha)$ 
\State $\beta \gets F_\alpha g_0$
\For{$i=1$ to $n-1$} 
	\State $F_\alpha\gets\alpha F_\alpha - V[i-1]P_\alpha$ \label{eval:updateFia} 
	\State $\beta\gets\beta + F_\alpha g_i$ \label{eval:updateBeta}
\EndFor
\State \Return $\beta$
\end{algorithmic}
\end{algorithm}

\begin{thm} \label{thm:evalmodPdense} 
  Algorithm~\ref{algo:eval} is correct. 
  It uses $\gO(n\#P)$ operations in $\gr$ and $\gO(n)$ operations in $\grext$.
\end{thm}

\begin{proof}
Step~\ref{eval:updateFia} relies on Equation~\eqref{eq:rec} to compute $F_\alpha = \F[i](\alpha)$. 
Step~\ref{eval:updateBeta} uses this evaluation together with Equation~\eqref{eq:eval} to correctly compute $(FG\bmod P)(\alpha)$.
The first step requires $\gO(n\#P)$ operations in $\gr$ by Lemma~\ref{lem:lastcoefdense}. (It does not depend on $\alpha$.) The other steps require $\gO(n)$ operations in $\grext$.
\end{proof}

As before, we notice that the operations in $\grext$ are sometimes \emph{scalar} multiplications, that is multiplications of an
element of $\grext$ by an element of $\gr$. We provide an analysis that minimizes the number of non-scalar multiplications.

\begin{cor}
  Let $P$, $F$, $G$ and $\alpha$ as in Algorithm~\ref{algo:eval}. Then $(FG)\bmod P$ can be evaluated on $\alpha$ using $2n-2$
  multiplications and $(3n-5+\#P)$ additions in $\grext$, $(3n-3+\#P)$ scalar multiplications in $\grext$, and $(n-1)(\#P-1)$ multiplications and
  additions in $\gr$.
\end{cor}

\begin{proof}
  We first note that the number of operations performed by Algorithm~\ref{algo:coef} is at most $(n-1)(\#P-1)$ multiplications and additions
  in $\gr$. In Algorithm~\ref{algo:eval}, we need to evaluate both $P$ and $F$ on $\alpha$. To minimize the number of non-scalar
  multiplications, we first compute $\alpha^2$, \dots, $\alpha^n$ using $n-1$ multiplications in $\grext$. We can then compute
  $P_\alpha$ using $\#P-1$ scalar multiplications and $\#P-1$ additions, and $F_\alpha$ using $n-1$ scalar multiplications and
  $n-2$ additions. Then, the initialization of $\beta$ and the for loop require $n-1$ multiplications, $2n-1$ scalar
  multiplications and $2n-2$ additions. This results in $2n-2$ multiplications in $\grext$, $(3n-3+\#P)$ scalar multiplications in
  $\grext$, and $(3n-5+\#P)$ additions in $\grext$.
\end{proof}

In a different context where the aim is specifically to compute the evaluation with no restriction to the use of polynomial arithmetic one can first compute the polynomial $FG\bmod P$ and then evaluate it on $\alpha\in\grext$.
Such method requires $\gO(n\log n\log\log n)$ operations in $\gr$ for the polynomial multiplication $F\times G$ and division by $P$ and $\gO(n)$ operations in $\grext$ for the evaluation.
Thus we see that if $P$ verifies $\#P<\log n\log\log n$ our technique is more efficient.

\subsection{Evaluation of a sparse modular product}\label{ssec:sparseval}

In this section, we adapt and analyse the previous algorithms for polynomials $F$ and $G$ given in sparse representation. 
Our results depend on the difference between the highest and the second highest exponents in $P$. Recall that the \emph{gap
  parameter} $\gamma$ is a measure of this difference, defined by $1-\gamma = \frac{1}{n}\max\{k < n:p_k\neq 0\}$. In particular,
the second highest exponent with nonzero coefficient in $P$ is
$(1-\gamma)n$. 
Proposition~\ref{prop:sparsity+coeffbound} gives a relation between the gap parameter and the sparsity of $(FG)\bmod
P$. The potential growth of the sparsity induced by the reduction modulo
$P$ explains the dependency of our results on the gap parameter.

As $G$ is sparse, Equation~\eqref{eq:eval} becomes 
\begin{equation}\label{eq:seval}
(FG\bmod P)(\alpha)=\sum_{i\in\supp(G)}g_i\F[i](\alpha),
\end{equation}
with the same notations as in the previous section.  The recurrence relation 
$\F[i+1]=X\F[i]-\f[i]_{n-1}P$ still holds, hence $\F[i+1](\alpha) = \alpha\F[i](\alpha) -
\f[i]_{n-1}P(\alpha)$ too. The goal now is to efficiently compute $\F[i](\alpha)$ for all
$i\in\supp(G)$ only, not for all indices $i$.
When $\gamma$ is not close to zero, there are actually few indices $i$ such that $\f[i]_{n-1}\neq
0$. In fact, the number of such indices depends on $\#F, \#P$ and $\gamma$.  Let $I =
\{i_1,\dotsc,i_t\}$ denote this set of indices. We will prove in Lemma \ref{lem:lastcoefsparse} that this set is of size
$\gO(\#F\#P^{\lceil 1/\gamma-1\rceil})$.  We decide to first provide some arguments and an explicit algorithm to prove this claim.

An important remark is that for any $0\leq j < n-1$, in particular those verifying $j\in\supp(G)$, if we assume $i$ to be the largest
index in $I$ not larger than $j$, then Equation~\eqref{eq:recFi} implies
\begin{equation}\label{eq:srec2}
  \F[j](\alpha) = \alpha^{j-i}\F[i](\alpha).
\end{equation}

Therefore, the recurrence relation given in~\eqref{eq:rec} becomes 
\begin{equation}\label{eq:srec1}
  \begin{cases}
    \F[i_{k+1}](\alpha) = \alpha^{i_{k+1}-i_k-1}(\alpha\F[i_k](\alpha) - \f[i_k]_{n-1}P(\alpha)) & \text{ for } k\geq 0\\
    \F[i_1](\alpha) = \alpha^{i_1}\F[0](\alpha)
  \end{cases}
\end{equation}
To efficiently use Equations~\eqref{eq:srec2} and~\eqref{eq:srec1} to perform the evaluation, we need to provide a sparse variant
of Algorithm~\ref{algo:coef}. It computes a sparse representation of the vector $V = [\f[0]_{n-1}, \dotsc, \f[n-2]_{n-1}]$, that is the
sparse vector $\{(i, \f[i]_{n-1}) : \f[i]_{n-1}\neq 0\}$.

The idea of Algorithm~\ref{algo:scoef} is to mimic Algorithm~\ref{algo:coef} in the sparse settings.  For simplicity of the presentation, we first
consider to store $V$ as a sparse vector as it is sufficient to our needs for proving our claims on the size of the set $I$.  We
will show in Corollary~\ref{cor:vanEmdeBoas} that we must require another structure to minimize the complexity attached to data
management.

The initial nonzero values in $V$ are the nonzero coefficients of $F = \F[0]$, with $V[i] = f_{n-1-i}$ if
$n-i-1\in\supp(F)$. 
Let now consider the external loop in Algorithm~\ref{algo:coef}. Iteration $i$ does not require any operation if $\f[i]_{n-1} = 0$ since
Equation~\eqref{eq:recFi} reduces to $\f[i+1]_k = \f[i]_{k-1}$ in that case. Therefore, we must loop over indices $i$ such that
$\f[i]_{n-1}$ is nonzero. For such an index $i$, the same updates as in Step~\ref{coef:update} of Algorithm~\ref{algo:coef} are
required. 
For $k\in\supp(P)$, $i < k < n$, we must perform the update $V[i+n-k]\gets V[i+n-k]-p_k\f[i]_{n-1}$. If $V[i+n-k]$ is already
nonzero, its value is already stored in $V$ and must be updated. Otherwise, the new value $-p_k\f[i]_{n-1}$ must be inserted in
$V$ with index $i+n-k$.

It remains to be able to only loop over the indices $i$ such that $\f[i]_{n-1} \neq 0$. Let us assume that iteration $i$ has been
performed since $\f[i]_{n-1} \neq 0$. The proof of Lemma~\ref{lem:lastcoefdense} shows that $V[i+1]$ then contains
$\f[i+1]_{n-1}$. Therefore in the sparse setting, we know that iteration $i+1$ has to be performed if, and only if,
$V[i+1]\neq 0$. More generally, the next index to be considered is the index of the next nonzero entry of $V$ after $V[i]$. 
Algorithm~\ref{algo:scoef}  below uses such method for computing all the indices $i$ such that $\f[i]_{n-1}$ is nonzero.

\begin{algorithm}
\caption{\scoef}
\label{algo:scoef}
\begin{algorithmic}[1]
  \Input $P$, $F\in\gr[X]$ with $\deg(F) < \deg(P)=n$, and $P$ monic 
 
  \Output The list $\{(i, \f[i]_{n-1}) : 0\le i < n-1, \f[i]_{n-1}\neq 0\}$, sorted by increasing values of $i$. 

\State $L\gets$ empty list
\State $V\gets\{(i,f_{n-1-i}) : n-1-i\in\supp(F)\}$ (sparse vector) 
\While{$V$ is not empty}
    \State $(i,v)\gets$ extract the element of minimal index from $V$
    \If{$v\neq 0$}
        \State Add $(i,v)$ to the list $L$ 
        \For{$k\in\supp(P)$ such that $i < k < n$} 
              \State $V[i+n-k]\gets V[i+n-k]-p_k v$ \label{scoef:update}
        \EndFor
    \EndIf
\EndWhile
\State \Return $L$
\end{algorithmic}
\end{algorithm}

\begin{lem} \label{lem:lastcoefsparse} \sloppy
  Algorithm~\ref{algo:scoef} is correct. If the polynomial $P$ has a gap parameter $\gamma$, the algorithm uses $\gO(\#F\#P^{\lceil
	1/\gamma-1\rceil})$ operations in $\gr$ and additions of $\gO(\log n)$-bits integers.
	In particular, there are at most $\#F\#P^{\lceil 1/\gamma-1\rceil}$ indices
  $i$ such that $\f[i]_{n-1}\neq 0$.
\end{lem}

\begin{proof}
  As explained above, Algorithm~\ref{algo:scoef} is an adaptation of Algorithm~\ref{algo:coef} to the sparse settings that only computes those
  $\f[i]_{n-1}$ that are nonzero, using Equations~\eqref{eq:srec2} and~\eqref{eq:srec1} in place of Equation~\eqref{eq:rec}. 
  Instead of considering all the $F_i[n-1]$ one after the other it only considers those which are not zero.
Let us call ``iteration $i$'' the iteration in the while loop that extract a pair $(i,v)$ from $V$.  To prove the correctness, we
prove by induction that at the end of iteration $i$,
\[
  V = \{(j, \f[i+1]_{n-j+i}) : j>i, \f[i+1]_{n-j+i}\neq 0\} \text{ and }
  L = \{(j, \f[j]_{n-1}) : j \le i, \f[j]_{n-1}\neq 0\}.
\]

  Before the loop (``iteration $-1$'') the property holds, that is $L$ is
empty and $V$ contains exactly the pairs $(j, \f[0]_{n-j-1})$ such that $\f[0]_{n-j-1}\neq 0$.

Let us assume that the property holds at then end of iteration $\ell$, and let $(i,v)$ be the pair extracted at the next iteration. We first prove that $\f[i]_{n-1} = v$ and $\f[j]_{n-1} = 0$ for $\ell < j < i$. By minimality of $i$ and induction hypothesis, $\f[\ell+1]_{n-j+\ell} = 0$ for $\ell < j < i$ at the end of iteration $\ell$. In particular, $\f[\ell+1]_{n-1} = 0$. By Equation~\eqref{eq:recFi}, $\f[\ell+2]_{n-j+\ell+1} = \f[\ell+1]_{n-j+\ell} - \f[\ell+1]_{n-1}p_{n-j+\ell+1} = 0$. And an easy recurrence shows that $\f[j]_{n-1} = 0$ for $\ell < j < i$. Now this implies that $\f[i]_{n-1} = \f[i-1]_{n-2} = \dotsb = \f[\ell+1]_{n-i+\ell}$. Yet by induction hypothesis, at the end of iteration $\ell$, $V$ contains $(j, \f[\ell+1]_{n-j+i})$ if $\f[\ell+1]_{n-j+i}\neq 0$. Therefore, if $\f[i]_{n-1}$ is nonzero, $\f[\ell+1]_{n-i+\ell}$ is nonzero too and $V$ contains the pair $(i, \f[\ell+1]_{n-i+\ell}) = (i, \f[i]_{n-1})$. In other words, the value $v$ extracted from $V$ is indeed equal to $\f[i]_{n-1}$ and the property holds for $L$ after iteration $i$.

Now with the same argument, $\f[i]_{n-1-j+i} = \f[\ell+1]_{n-j+\ell}$ for $\ell < j < i$. Right before iteration $i$, $V$ contains then the pairs $(j, \f[i]_{n-1-j+i})$ for $\f[i]_{n-1-j+i}\neq 0$. After iteration $i$, such pairs are replaced by $(j, \f[i]_{n-1-j+i} - p_{n-j+i}\f[i]_{n-1})$, that is $(j, \f[i+1]_{n-j+i})$ by Equation~\eqref{eq:recFi}. And if $\f[i]_{n-1-j+i} = 0$ but $p_{n-j+i}\neq 0$, a new pair $(j, -p_{n-j+i}\f[i]_{n-1}) = (j, \f[i]_{n-j+i})$ is inserted into $V$. Therefore, the property holds for $V$ too after iteration $i$.

The second point is to count the number of operations.  Since the while loop stops when $V$ is empty, the number of operations in $\gr$ is
at most twice the number of pairs that are inserted into $V$ during the algorithm
	and the same number of additions in $\gz$ are performed as the index of each pair is computed by two additions of number at most $n$. We will classify the pairs by
\emph{generations}. Initially, $V$ contains $\#F$ pairs which form generation $0$. New pairs can be inserted into $V$ when a pair
$(i,v)$ is extracted. If $(i,v)$ is a pair of generation $t$, the new pairs inserted at iteration $i$ belong to generation
$t+1$. At any iteration, at most $\#P-1$ pairs are inserted into $V$. Therefore, there are at most $\#F(\#P-1)$ pairs of
generation $1$, $\#F(\#P-1)^2$ pairs of generation $2$, and in general $\#F(\#P-1)^t$ pairs of generation $t$. Now we need to
bound the number of generations. Note the pairs of generation $0$ have an index $i$ between $0$ and $n-1$. But at generation $1$,
the new pairs have index $(i+n-k)$ for some $k\in\supp(P)$, $k < n$. There comes the gap into account: If $P$ has gap parameter
$\gamma$, the largest exponent less than $n$ in $\supp(P)$ is $(1-\gamma)n$ by definition. Therefore, at generation $1$, all pairs
have an index at least $i+n-(1-\gamma)n\ge\gamma n$. At generation $2$, all pairs have then an index at least $2\gamma n$. At
generation $t$, all pairs have an index at least $t\gamma n$. Since indices are bounded by $n-1$, there cannot be any pair of
generation $t$ if $t\gamma n\ge n$. In other words, the largest possible generation is $t = \lceil 1/\gamma -
1\rceil$. Altogether, the total number of pairs inserted into $V$ is at most
\[\sum_{t=0}^{\lceil 1/\gamma-1\rceil} \#F(\#P-1)^t 
        = \#F\ \frac{(\#P-1)^{\lceil 1/\gamma\rceil}-1}{\#P-2} \]
if $\#P > 2$, and is at most $\lceil {1}/{\gamma}\rceil\#F$ if $\#P = 2$. 
To simplify the exposition, we bound both of them by $\#F\#P^{\lceil {1}/{\gamma}-1\rceil}$ in the following.
Note that of course, this number is also a bound on the number of pairs that are extracted from $V$ during the algorithm.

This has two consequences. First, the number of extracted pairs is a bound on the size of the list at then end of the algorithm. Therefore there are at most $\gO(\#F\#P^{\lceil 1/\gamma-1\rceil})$ nonzero values $\f[i]_{n-1}$. Second, this number also bounds the total number of executions of Step~\ref{scoef:update}, that is the total number of operations. 
\end{proof}

\begin{cor}\label{cor:vanEmdeBoas}
  All operations of \textsc{Insertion}, \textsc{Removal}, \textsc{Minimum} and \textsc{Search} of pairs $(i,\nu)$ in the data
  structure $V$ within Algorithm~\ref{algo:scoef} can be done with $\gO(\#F\#P^{\lceil 1/\gamma-1\rceil}\log\log n)$ bit operations.
\end{cor}
\begin{proof}
  By definition the size of the sparse vector $V$ is at most $n$. Therefore, using a data structure for $V$ of type \emph{van Emde
    Boas} tree with a universe of size $n$, ensures that any requested operations can be done with $\gO(\log \log n)$ bit
  operations, see \cite[Chapter 20]{CormenBook}.
\end{proof}

\begin{rem}
	As the bit-complexity of all the operations in $\gz$ of Algorithm~\ref{algo:scoef} is $\gO(\#F\#P^{\lceil 1/\gamma-1\rceil}\log n)$ the cost driven by the data structure of $V$ is negligible.
\end{rem}

Our algorithm to compute the evaluation of the polynomial $FG\bmod P$ on some point $\alpha$, when $F$ and $G$ are given in sparse
representation, relies on Equations~\eqref{eq:seval}, \eqref{eq:srec1} and~\eqref{eq:srec2}.  More precisely, we first compute
each $\F[i](\alpha)$ for indices $i$ such that $\f[i]_{n-1}\neq 0$ by means of Equation~\eqref{eq:srec1}. From these values, we
get each $\F[j](\alpha)$ for $j\in\supp(G)$ by means of~Equation~\eqref{eq:srec2}. Finally, we deduce $(FG\bmod P)(\alpha)$ using
Equation~\eqref{eq:seval}.

In Algorithm~\ref{algo:seval}, all these computations are intertwined.  The idea is to loop over all indices $j$ such that either
$\f[j]_{n-1}\neq 0$ or $j\in\supp(G)$. If $\f[j]_{n-1}\neq 0$, we update the value $\F[j](\alpha)$ using
Equation~\eqref{eq:srec1}. If $j\in\supp(G)$, we accumulate partial evaluations of $(FG\bmod P)(\alpha)$ using
Equations~\eqref{eq:seval} and~\eqref{eq:srec2}.

\begin{algorithm}
\caption{\seval}
\label{algo:seval}
\begin{algorithmic}[1]
  \Input $P$, $F$ and $G\in\gr[X]$, with $\deg(F),\deg(G) < \deg(P)=n$, $P$ monic and $\alpha\in\grext$.
 
  \Output $(FG\bmod P)(\alpha)$

\State $V\gets\{(i,\f[i]_{n-1}): 1\le i < n-1, \f[i]_{n-1}\neq 0\}$, using a call to $\scoef(P,F)$ \label{seval:scoef}
\If {$\f[0]_{n-1} = 0$} insert $(0,0)$ in $V$\EndIf
\State $P_\alpha\gets P(\alpha)$ 
\State $F_\alpha\gets F(\alpha)$
\State $\beta\gets F_\alpha g_0$ \Comment $\beta\gets0$ if $0\notin\supp(G)$
\State $i\gets 0$
\For{$j\in\supp(V)\cup\supp(G)\setminus\{0\}$, by increasing order} \label{seval:loop}
    \If{$j\in\supp(V)$}
      \State $F_\alpha\gets\alpha^{j-i-1}(\alpha F_\alpha - V[i] P_\alpha)$ \Comment Equation~\eqref{eq:srec1} \label{seval:Fa}
      \State $i\gets j$
    \EndIf
    \If{$j\in\supp(G)$}
      \State $\beta\gets\beta + \alpha^{j-i}F_\alpha g_j$ \Comment Equations~\eqref{eq:seval} and \eqref{eq:srec2} \label{seval:beta}
    \EndIf
\EndFor
\State \Return $\beta$
\end{algorithmic}
\end{algorithm}

\begin{thm} \label{thm:evalmodPsparse} 
Algorithm~\ref{algo:seval} is correct. 
It uses $\gO(\#F\#P^{\lceil\frac{1}{\gamma}-1\rceil})$ operations in $\gr$,
	$\gO((\#F\#P^{\lceil\frac{1}{\gamma}-1\rceil}+\#G)\log n)$ operations in $\grext$.
\end{thm}

\begin{proof}
  We prove that at the end of iteration $j$, $F_\alpha = \F[j](\alpha)$ if $j\in\supp(F)$ and $\beta = \sum_i g_i\F[i](\alpha)$
  where the sum ranges over indices $i\in\supp(G)\cap\{0,\dotsc,j\}$. The property is satisfied after iteration $0$ (before
  entering the loop) since $F_\alpha = \F[0](\alpha)$ and $\beta = F_\alpha g_0 = g_0\F[0](\alpha)$. Let us assume that the
  property holds before entering iteration $j$. Index $i$ denotes the previous index that belongs to $\supp(F)$. Therefore, if
  $j\in\supp(F)$, Equation~\eqref{eq:srec1} ensures that $F_\alpha$ has the right value after iteration $j$ since
  $V[i] = \f[i]_{n-1}$. And Equations~\eqref{eq:seval} and~\eqref{eq:srec2} justify that $\beta$ also has the right value if
  $j\in\supp(G)$.

  The evaluations $P(\alpha)$ and $F(\alpha)$ require $\gO(\#P\log n)$ and $\gO(\#F\log n)$ operations in $\grext$
  respectively. Steps~\ref{seval:Fa} and~\ref{seval:beta} each require $\gO(\log n)$ operations in $\grext$ to compute powers of
	$\alpha$ and $\gO(1)$ additions with integers of size $\gO(\log n)$ to compute the appropriate exponent.
	These steps are executed $\gO(\#V+\#G)$ times. Since $\#V = \gO(\#F\#P^{\lceil 1/\gamma-1\rceil})$ this gives a total
	of $\gO((\#F\#P^{\lceil\frac{1}{\gamma}-1\rceil}+\#G)\log n)$ operations in $\grext$ plus $\gO((\#F\#P^{\lceil\frac{1}{\gamma}-1\rceil}+\#G)\log n)$ bit operations for the integer additions.
	Since we can easily assume that one operation in $\grext$ will cost more than one bit operation, the latter complexity is dominated by the computation part in $\grext$.
	The cost of Step~\ref{seval:scoef} is given by Lemma~\ref{lem:lastcoefsparse}
	
	. We can still use \emph{van Edme Boas} tree to iterate
  over the union of the supports of $V$ and $G$ at Step~\ref{seval:loop} with a total of
	$\gO((\#F\#P^{\lceil\frac{1}{\gamma}-1\rceil}+\#G)\log\log n)$ bit operations which is less than the number bit operations required by the additions in $\gz$ and thus negligible. 
\end{proof}

Obviously, since polynomial multiplication over integral domains is commutative, the roles of $F$ and $G$ can be exchanged in Algorithm~\ref{algo:seval}. In particular if $\#G < \#F$, this exchange decreases the complexity in Theorem~\ref{thm:evalmodPsparse}. In other words, the statement remains valid if $\#F$ is replaced by $\min(\#F,\#G)$ and $\#G$ by $\max(\#F, \#G)$. The same remark applies to subsequent results.

\begin{rem}
  As in Corollary~\ref{cor:sevalbin}, we can decrease the number of operations over $\grext$ by using simultaneous exponentiation
  on $\alpha$. This results in $\log n + \gO((\#F\#P^{\lceil 1/\gamma-1\rceil}+\#G)\log n/\log\log n)$ operations in $\grext$.
\end{rem}

If the gap parameter $\gamma$ is close to $\frac{1}{n}$, the polynomial $FG\bmod P$ is in general a dense polynomial even if $FG$ is sparse and the dense modular evaluation will be more appropriate.
On the contrary, $FG\bmod P$ remains sparse if $\gamma$ is close to $1$, in particular if $\gamma\geq\frac{1}{2}$.

\begin{rem}
	If $\gamma\geq\frac{1}{2}$, the evaluation requires $\gO((\#F\#P+\#G)\log n)$ operations in $\grext$.
\end{rem}

\begin{rem}
  The factor $\#F\#P^{\lceil\frac{1}{\gamma}-1\rceil}+\#G$ in the complexity may be larger than the actual sparsity of
  $FG\bmod P$. For instance, the sparsity can be $0$ if $P$ divides $FG$. Yet, it is smaller than the general bound
  $\#F\#G(\#P-1)^{\lceil\frac{1}{\gamma}\rceil}$ given by Proposition~\ref{prop:sparsity+coeffbound}.
	Thus in general it is more efficient to use our method than to directly evaluate $FG\bmod P$ if the polynomial is known.
\end{rem}

\section{Verification of polynomial modular product}\label{sec:verif}

This section is devoted to the verification of polynomial modular product. That is, given $F$, $G$, $H$ and $P$, such
that $\deg(F),\deg(G),\deg(H)<\deg(P)=n$, we want to test whether $H = FG\bmod P$. The idea is classical, that is to evaluate the
identity at a random point.  Contrary to the more straightforward verification of polynomial multiplication, we cannot do such
evaluation directly since we do not know the polynomial $Q= (FG - (FG \bmod P))/P$.  As seen in the previous section, we provide
new algorithms to do such evaluation efficiently without reverting to the computation of $Q$.  We remind that $P$ is always
taken monic.
Note that this is a mild assumption since $(FG)\bmod P = (FG)\bmod(\lambda P)$ for any invertible
constant $\lambda$.

In the following, all algorithms are analysed both when the polynomials $F$, $G$ and $H$ are dense and when they are sparse. On the other hand, $P$ is always considered as a sparse polynomial. We shall recall
that $\gamma$ denotes the gap parameter of $P$, defined by $1-\gamma = \deg(P-X^n)$, and it serves to control the densification of
the modular reduction.

We first begin in Section~\ref{ssec:verifIntDom} with an abstract case where the polynomials are defined over an integral domain.
There, we analyse the algorithms by counting the number of ring operations.  In Sections~\ref{ssec:verifoverZ}
and~\ref{ssec:smallfieldverif} we discuss some adaptations of the algorithm to the case of integers and small finite fields in
order to provide finer analysis in the bit complexity model.

\subsection{Modular product verification in $\gr[X]$}\label{ssec:verifIntDom}

Algorithm~\ref{algo:verify} depicted below is straightforward from Theorems~\ref{thm:evalmodPdense} or~\ref{thm:evalmodPsparse}. We
mainly provide his description to serve as a starting point for its adaptations in the next sections. The algorithm covers both the
dense and the sparse case. The only difference is at Step~\ref{verify:sparse}.

\begin{algorithm}
\caption{\verify}
\label{algo:verify}
\begin{algorithmic}[1]
  \Input $F$, $G$, $H$ and $P\in\gr[X]$, $P$ monic of degree $n$ and $F$, $G$ and $H$ of degrees $< n$ ; $0<\epsilon<1$.
 
  \Output True if $H = FG\bmod P$, False with probability at least $1-\epsilon$ otherwise.

\If{$F$, $G$ and $H$ are given in sparse representation} \label{verify:sparse}
    \If{$\#H > \#F\#G(\#P-1)^{\lceil 1/\gamma\rceil}$} 
	\Return False
	\Comment Proposition~\ref{prop:sparsity+coeffbound}
    \EndIf
\EndIf
\State $\alpha\gets$ random element from a subset $\mathcal{E}$ of $\gr$, of size $\ge\frac{1}{\epsilon}(n-1)$ 
\State $\beta \gets (FG\bmod P)(\alpha)$ \Comment using Theorem~\ref{thm:evalmodPdense} or~\ref{thm:evalmodPsparse} 
\State \Return True if $\beta= H(\alpha)$, False otherwise
\end{algorithmic}
\end{algorithm}

\begin{thm}\label{thm:verifmod}
If $\gr$ has at least $\frac{1}{\epsilon}(n-1)$ elements, Algorithm~\ref{algo:verify} is correct. 

If $F$, $G$ and $H$ are dense, the algorithm uses $\gO(\#P n)$ operations in $\gr$.

If $F$, $G$ and $H$ are sparse, the algorithm uses $\gO((\#F\#P^{\lceil\frac{1}{\gamma}-1\rceil}+\#G+\#H)\log n)$ operations in $\gr$.
\end{thm}

\begin{proof}
  Step~1 dismisses a trivial mistake if the polynomials are sparse. If $H = (FG)\bmod P$, $H(\alpha) = (FG\bmod P)(\alpha)$ for
  any $\alpha$ and the algorithm always returns True. Otherwise, let $\Delta = H-(FG)\bmod P$. Then $\Delta$ has degree $< n$,
  hence has at most $n-1$ roots since it is nonzero. Therefore, the probability that $\alpha$, randomly chosen in $\mathcal E$, is
  a root of $\Delta$ is at most $(n-1)/\frac{1}{\epsilon} (n-1) = \epsilon$. The algorithm returns True in that case with
  probability at most $\epsilon$.

  The complexity is given by the cost of a single modular product evaluation that is stated in Theorem~\ref{thm:evalmodPdense} for
  the dense case and in Theorem~\ref{thm:evalmodPsparse} for the sparse case.
\end{proof}

If $\gr$ is not large enough, the algorithm fails and it is customary to revert to an extension ring $\grext$ to perform the
evaluation in a larger set.  Using Theorems~\ref{thm:evalmodPdense} and~\ref{thm:evalmodPsparse}, we get the following extension
when the polynomials are evaluated on a random point of $\grext$ rather than $\gr$.

\begin{cor}\label{cor:verifmodext}
  Let $\gr$ be an integral domain with less than $\frac{1}{\epsilon}(n-1)$ elements, and $\grext$ an extension ring of $\gr$ with
  at least $\frac{1}{\epsilon}(n-1)$ elements. Then Algorithm~\ref{algo:verify} can be adapted by choosing a random element from $\grext$, with the
  same probability of success. It uses $\gO(n\#P)$ operations in $\gr$ and $\gO(n)$ operations in $\grext$ if $F$, $G$ and $H$ are
  dense, and $\gO((\#F\#P^{\lceil\frac{1}{\gamma}-1\rceil}+\#G+\#H)\log n)$ operations in $\grext$ if they are sparse.
\end{cor}

In the dense case, Algorithm~\ref{algo:verify} uses an optimal 
number of operations in $\gr$ as soon as $\#P$ is constant. It  
is always faster than a modular product when $\#Pn<\mul(n)$ that is when $\#P<\log (n)\log\log(n)$ for a general ring $\gr$.  In
the sparse case, Theorem~\ref{thm:verifmod} is not linear in the input size. Indeed, $\#P$ is raised to a potentially large power
$\lceil\frac{1}{\gamma}-1\rceil$ and more importantly there is a factor $\log n$ in the number of operations in $\gr$ while the input
has only $(\#F+\#G+\#H)$ elements of $\gr$.
Nevertheless, the efficiency of verification has to be compared with the cost of computing $FG\bmod P$. We assume the latter to be
done with a sparse multiplication followed by a sparse division with $P$. This hypothesis seems reasonable as no work have been
done to optimize such operation yet. Letting aside the division, the number of operations in $\gr$ for sparse multiplication could
be either $\gO(\#F\#G)$ with naive approach or $\gOt(\#F+\#G+\#(FG))$ using \cite{giorgi2020:sparsemul}. Assuming $\#P$ to be constant,
our verification has a complexity of $\gO((\#F+\#G+\#H)\log n)$ which is always faster when $\#(FG)=o(\#F\#G/\log n)$.  If we
assume that $\#(FG)=\gO(\#F\#G)$, our verification will be faster at least when $n = 2^{\gO(\#F+\#G)}$ and $\#P$ constant.
Depending on the cost of the division, our algorithm could be faster in more cases.

Of course, these conditions are not very restrictive. We use Algorithm~\ref{algo:verify} in Section~\ref{sec:product} to
verify classical polynomial multiplication, where $P$ will be a binomial of degree either logarithmic in the sparsity or
polynomial in the input degree.

Yet the efficiency of Algorithm~\ref{algo:verify} depends heavily on the integral domain $\gr$.  Indeed 
the complexity
of polynomial multiplication in $\gr[X]$ can be faster than $\gO(n\log n\log\log n)$ operations in $\gr$
\cite{vdH:cyclomult,harvey:polmul-nlogn}. Furthermore, if $\gr$ is small, one operation in $\grext$ corresponds to a non-constant
number of operations in $\gr$.  In the following sections we consider polynomials over the integers or finite fields and we
provide thorough analyses together with adapted versions when necessary.

\subsection{Modular product verification in  $\gz[X]$}\label{ssec:verifoverZ}

If the polynomials are defined over $\gz$, there is no difficulty with the size of $\mathcal{E}$ in
Algorithm~\ref{algo:verify}. However we must prevent the integers growth during the evaluation. It is very classical to choose a
random prime $q$ and to map the whole computation into $\cfq$. To do so, we must ensure two properties on the prime $q$. First,
$q$ must be large enough to use the algorithm, that is at least $\frac{1}{\epsilon}(n-1)$. Second, if $H\neq (FG)\bmod P$, we need
this inequality to hold modulo $q$ as well. For this second property, we define $\Delta = H-(FG)\bmod P$. To ensure that $\Delta$
does not vanish modulo $q$, we need that at least one coefficient of $\Delta$ is nonzero modulo $q$. We then need to bound its
coefficients to assess the latter fact.

\begin{prop}\label{prop:productnorm}
	The coefficients of $\Delta$ are bounded  by
	$\norm{H} + \min(\#F,\#G)\norm{F}\norm{G}(\#P\norm{P})^{\lceil\frac{1}{\gamma}\rceil}$. 
\end{prop}

\begin{proof}
	The coefficient of $\Delta$ are bounded by $\norm{H}+\norm{FG\bmod P}$.
	The bound follows from Proposition~\ref{prop:sparsity+coeffbound}, with $Q = FG$ and $\norm Q\le \min(\#F,\#G)\norm F\norm G$ by Lemma~\ref{lem:prodbounds}.
\end{proof}

Using this bound and Proposition~\ref{prop:choixdeq} we can determine an appropriate prime $q$ to adapt the
Algorithm~\ref{algo:verify} to the integer case. This is done in the Algorithm~\verifyZ below.

\begin{algorithm}
\caption{\verifyZ}
\label{algo:verifyZ}
\begin{algorithmic}[1]
  \Input $F$, $G$, $H$ and $P\in\gz[X]$, $P$ monic of degree $n$ and $F$, $G$ and $H$ of degrees $< n$ ; $0<\epsilon<1$.
 
  \Output True if $H = FG\bmod P$, False with probability at least $1-\epsilon$ otherwise.

	\State $\Delta_\infty\gets \norm{H} + \min(\#F,\#G)\norm{F}\norm{G}(\#P\norm{P})^{\lceil\frac{1}{\gamma}\rceil}$ 
\State $q\gets \rp(\lambda, \frac{\epsilon}{2})$ where $\lambda = \max(21, \frac{2}{\epsilon} n, \frac{20}{3\epsilon}\ln\Delta_\infty)$ \label{verifZ:q}
\State $(F_q,G_q,H_q,P_q)\gets (F\bmod q, G\bmod q, H\bmod q, P\bmod q)$ \label{verifZ:reduction}
\State \Return $\verify(F_q,G_q,H_q,P_q,\frac{\epsilon}{2})$ \label{verifZ:verify}
\end{algorithmic}
\end{algorithm}

\begin{thm}\label{thm:sparsemodverifZ}
Algorithm~\ref{algo:verifyZ} is correct. 
If $C = \max(\norm P, \norm F, \norm G, \norm H)$ and $T = \max(\#F,\#G,\#H)$, the algorithm requires
 $\gO(\log^2(\frac{n}{\epsilon}\log C)\log\log(\frac{n}{\epsilon}\log C)\log\frac{1}{\epsilon}\imul(\log(\frac{n}{\epsilon}\log C)))$ bit operations to get a prime number, plus

\begin{itemize}
	\item $\gO((\#Pn+n\frac{\log C}{\log(\frac{n}{\epsilon}\log C)})\imul(\log(\frac{n}{\epsilon}\log C)))$ bit operations if $F$, $G$ and $H$ are dense, or

	\item $\gO((T\#P^{\lceil\frac{1}{\gamma}-1\rceil}\log n+(T+\#P)\frac{\log C}{\log(\frac{n}{\epsilon}\log C)})\imul(\log(\frac{n}{\epsilon}\log C)))$ bit operations if $F$, $G$ and $H$ are sparse. 
\end{itemize}
\end{thm}

\begin{proof}
  To ensure that \verify works properly, we need $q$ to be at least $\frac{2}{\epsilon} (n-1)$. This is the case since
  $\lambda\ge \frac{2}{\epsilon}n$. The algorithm always returns the correct answer when $H = (FG)\bmod P$. Otherwise, it may
  incorrectly return True in two cases: Either $H_q = (F_qG_q)\bmod P_q$ while the equality does not hold over $\gz$, or \verify
  incorrectly returns True. Both situations occur with probability at most $\frac{\epsilon}{2}$.  Indeed, by
  Proposition~\ref{prop:productnorm}, $\Delta_\infty$ is always a bound on $\norm\Delta$ where $\Delta = H-(FG)\bmod
  P$. Therefore, Proposition~\ref{prop:choixdeq} shows that with probability at least $\frac{\epsilon}{2}$, the number $q$ chosen
  at Step~\ref{verifZ:q} is a prime number such that $\Delta\bmod q\neq 0$. The error probability of Algorithm~\ref{algo:verifyZ} is thus at most
  $\epsilon$.

  Let us now analyse the complexity of the algorithm. As a first step, we shall express $q$ in terms of the input size.
	Since $\#P$, $\#F$, $\#G \le n$,  $\Delta_\infty = \gO(n^{1+\lceil\frac{1}{\gamma}\rceil}C^{2+\lceil\frac{1}{\gamma}\rceil})$. 
	Thus,  $\ln\Delta_\infty = \gO(\frac{1}{\gamma}(\log n+\log C)) = \gO(n(\log n+\log C))$ since $\frac{1}{\gamma}\le n$. This implies
  $q = \gO(\frac{n}{\epsilon}(\log n + \log C))$ and $\log q = \gO(\log(\frac{n}{\epsilon}\log C))$.

	By Proposition~\ref{prop:rdprime}, Step~\ref{verifZ:q} costs $\gO(\log\frac{1}\epsilon\log^2q\log\log q \imul(\log q))$ bit
  operations, that is
\begin{equation}\label{eq:findq}
\gO\left(\log\tfrac{1}{\epsilon} \log^2(\tfrac{n}{\epsilon}\log C) \log\log(\tfrac{n}{\epsilon}\log C) \imul(\log(\tfrac{n}{\epsilon}\log C))\right).
\end{equation}
Step~\ref{verifZ:reduction} requires $\gO(n\frac{\log C}{\log q}\imul(\log q))$ bit operations in the dense case, and
$\gO((T+\#P)\frac{\log C}{\log q}\imul(\log q))$ bit operations in the sparse case. By Theorem~\ref{thm:sparsemodverifZ},
Step~\ref{verifZ:verify} requires $\gO(\#Pn\imul(\log q))$ bit operations in the dense case and
$\gO(T\#P^{\lceil\frac{1}{\gamma}-1\rceil}\log n\imul(\log q))$ bit operations in the sparse case.  Adding the complexities of all
these steps leads to the claimed bit complexity.
\end{proof} 

\begin{rem}\label{rem:negligible}
  In most cases, the cost of finding a prime number is negligible in comparison to the rest of the algorithm. 
  In the dense case, it is negligible as long as $\epsilon = 1/n^{\gO(1)}$. 
  In the sparse case, it is negligible when 
  the degree $n$ is not too large compared to the other input parameters. More precisely, this is the case when
  $\frac{n}{\epsilon}=((T+\#P)\log C)^{\gO(1)}$.
\end{rem}

When computing a multiplication followed by a modular reduction with polynomials in $\gz[X]$, the size of the coefficients can
grow significantly as shown by the bound on $\norm{FG\bmod P}$ given in the proof of Proposition~\ref{prop:productnorm}.  At the
opposite, our verification algorithm is done with bounded integers of bit length $\gO(\log(\frac{n}{\epsilon}\log C))$.  This
is logarithmic in the input size in dense representation and linear in the sparse one.  Our verification therefore 
avoids paying the coefficient growth, contrary to the direct computation. 
Taking this growth into account to compare our verification algorithm with the computation seems hard.
We only detail the cases where our algorithm is already
faster even without considering the coefficient growth. We shall mention that for sparse polynomial, $\#(FG \bmod P)$ might be smaller that $\#H$. 
In our analysis, we assume for simplicity that both have approximately the same size.

\begin{rem}
	
        For $\epsilon = 1/n^{\gO(1)}$,
        Algorithm~\ref{algo:verifyZ} is faster than the polynomial modular product 
	\begin{itemize}
		\item[(i)] when the polynomials are dense and $\#P< \min(\frac{\log n}{\log\log n},\frac{\log C}{\log\log\log C})$;
		\item[(ii)] when the polynomials are sparse and $\frac{n}{\epsilon}=((T+\#P)\log C)^{\gO(1)}$.
                \end{itemize}
\end{rem}

\begin{proof}
    We assume $\epsilon = 1/n^{\gO(1)}$. In particular, $\log\frac{n}{\epsilon} = \gO(\log n)$.
  To simplify the analysis, we place ourselves in the case of a negligible cost for finding the prime $q$, as described in Remark~\ref{rem:negligible}. 
  In the sparse case, this implies that $\frac{n}{\epsilon}$ must be polynomial in $(T+\#P)\log C$.
  Therefore $\log n<\min(T,\#P)$
  and $\imul(\log(n\log C))=\imul(\log\log C)$. This makes our verification faster than computing the modular product since the
  latter requires at least $\#(FG\bmod P)=\gO(T^2\#P^{\lceil\frac{1}{\gamma}\rceil})$ operations on integers of bit length
  $\log C$. 

  For the dense case, we compare to the cost of multiplying two polynomials of degree $n$ and coefficients bounded by
  $C$.  As seen in the introduction, this reduces to integer multiplication with Kronecker substitution and it costs
  $\imul(n(\log n+\log C))=\gO(n(\log^2 n +\log n\log C + \log C\log\log C)$ bit operations. By
  Theorem~\ref{thm:sparsemodverifZ} our verification needs $\gO(\#Pn\imul(\log n+\log\log C) + n\log C\log(\log n+\log\log C))$
  bit operations. The second term is dominated by the complexity of multiplying the polynomials when $\epsilon = 1/n^{\gO(1)}$. The first term is 
  \[\gO(\#P n((\log n+\log\log C)\log\log n + (\log n + \log\log C)\log\log\log C)).\]
  When $\#P< \min(\frac{\log n}{\log\log n},\frac{\log C}{\log\log\log C})$, this term is bounded by 
  $\gO(n(\log^2 n +\log n\log C + \log C\log\log C))$, that is the complexity of computing the product.
\end{proof}

\subsection{Modular product verification in $\cfq[X]$}\label{ssec:smallfieldverif}
The situation over a finite field $\cfq$ is different since there is no growth to prevent.  When $q$ is large enough,
Theorem~\ref{thm:verifmod} applies directly.  Otherwise, one can revert to computing in a sufficiently large extension field of
$\cfq$, where Corollary~\ref{cor:verifmodext} can be applied.  We first give the precise complexity bounds for these two cases.

\begin{cor}\label{cor:smallfieldbyextension}
  Let $F$, $G$, $H$ and $P\in\cfq[X]$ as in Algorithm~\ref{algo:verify}. One can test whether $H = (FG)\bmod P$ using
  Algorithm~\ref{algo:verify}, with
  $\gO(\log\frac{1}{\epsilon}(\log_q \frac{n}{\epsilon})^2\mul(\log_q \frac{n}{\epsilon})(\log q+\log\log_q \frac{n}{\epsilon}))$
  operations in $\cfq$ to get an irreducible polynomial of degree $\gO(\log_q\frac{n}{\epsilon})$ , plus
\begin{itemize}
\item $\gO(n\#P+n\mul(\log_q\frac{n}{\epsilon}))$ operations in $\cfq$ if $F$, $G$ and $H$ are dense, or
\item $\gO(T\#P^{\lceil\frac{1}{\gamma}-1\rceil}\log n\mul(\log_q \frac{n}{\epsilon}))$ operations in $\cfq$ if $F$, $G$ and $H$ are sparse with at most $T$ nonzero monomials.
\end{itemize}
\end{cor}

\begin{proof}
  Let us assume that $q < \frac{1}{\epsilon}(n-1)$ otherwise Theorem~\ref{thm:verifmod} applies straightforwardly. In that
  case, Corollary~\ref{cor:verifmodext} requires to choose a random point in an extension $\cf_{q^d}$ of $\cfq$ with at least
  $\frac{1}{\epsilon}(n-1)$ elements. More precisely, we use Proposition~\ref{prop:irrpoly} to produce with probability
  $1-\frac{\epsilon}{2}$ an irreducible polynomial of degree $d$ over $\cfq$, where $d$ is the smallest integer such that
  $q^d\geq \frac{2}{\epsilon}(n-1)$. The algorithm may be incorrect if either the polynomial used to defined $\cf_{q^d}$ fails
  to be irreducible, or if the Algorithm~\ref{algo:verify} fails. If we choose $\alpha$ at random in $\cf_{q^d}$, the error probability of
  Algorithm~\ref{algo:verify} is at most $\frac{\epsilon}{2}$. This gives a total probability of error of at most $\epsilon$.
  
  By definition, the degree of the extension is $d=\gO(\log_q \frac{n}{\epsilon})$. The cost of generating the irreducible
  polynomial of degree $d$ is $\gO(\log\frac{1}{\epsilon}d^2\mul(d)(\log q+\log d))$ by
  Proposition~\ref{prop:irrpoly}. 
  Using Corollary~\ref{cor:verifmodext} and the fact that an operation in $\cf_{q^d}$ costs $\mul(d)$ operations in $\cfq$, we
  obtain the claimed costs.
\end{proof}

\begin{rem}\label{rem:verifcfepsilonfixed}
  If $\epsilon=1/n^{\gO(1)}$, the cost of getting an irreducible polynomial is negligible in the dense case.
  Then the algorithm requires $\gO(\#Pn+n\mul(\log_q n))$ operations in $\cfq$.  If we add the degree constraint $\log n=o(T)$,
  the cost of getting an irreducible polynomials is also negligible in the sparse case and the algorithm requires
  $\gO(T\#P^{\lceil\frac{1}{\gamma}-1\rceil}\log n\mul(\log_q n))$ operations in $\cfq$.
\end{rem}

As $\#(FG\bmod P)$ is bounded by $T^2\#P^{\lceil\frac{1}{\gamma}\rceil}$ and computing $FG\bmod P$ requires at least one operation
in $\cfq$ for each monomial, this remark gives directly a case where the verification is faster than the modular product in
general.  Moreover, naive algorithms can be used to perform products in the extension of $\cfq$.

\begin{rem}\label{rem:verifCostBetterSparse} 
  When $\epsilon=1/n^{\gO(1)}$ and $n =T^{\gO(1)}$, Algorithm~\ref{algo:verify} in the sparse case is in general
  faster than the modular multiplication and requires $\gOt(T\#P^{\lceil\frac{1}{\gamma}-1\rceil})$ bit operations if 
 $q<\frac{n}{\epsilon}$. 
\end{rem}

In the dense case, we can see from Remark~\ref{rem:verifcfepsilonfixed} that the verification complexity might in fact be larger than the cost of computing the modular product
$FG \bmod P$. Indeed, assuming $\mul_q(n)=\gO(n \log q(\log n+ \log \log q))$ bit operations \cite{harvey:polmul-nlogn}, we have
$\mul_q(\log_q\frac{n}{\epsilon})= \gO(\log \frac{n}{\epsilon} (\log \log \frac{n}{\epsilon} + \log \log q))$.  While the cost of
computing $FG \bmod P$ uses $\mul_q(n)$ bit operations, our verification requires
$\gO(\#Pn\log q\log\log q+ n \log \frac{n}{\epsilon}( \log \log \frac{n}{\epsilon}+ \log \log q))$ bit operations.  When $\#P$ is
not a constant, the latter is always larger.  The following remark precise when we can expect a positive result.

\begin{rem}\label{rem:verifCostBetter}
  Assuming $\#P$ to be constant and $\epsilon=1/n^{\gO(1)}$, Algorithm~\ref{algo:verify} in the dense case with $\gr=\cfq$ is
  asymptotically faster than the modular multiplication when 
  \begin{itemize}
  \item[(i)] $\log \frac{n}{\epsilon}< q < \frac{n}{\epsilon}$, since modular multiplication costs $\gO(n\log q \log n)$ bit operations
    while verification, which does need an extension, is $\gO(n\log \frac{n}{\epsilon} \log \log \frac{n}{\epsilon})$.
  \item[(ii)] $ {\frac{n}{\epsilon}}< q < 2^{\frac{n}{\epsilon}}$, since modular multiplication costs $\gO(n\log q \log n)$ bit operations while
    verification, which does not use an extension, is $\gO(n\log q \log \log q)$;
  \end{itemize}
\end{rem}

If the field is very large $q>2^{\frac{n}{\epsilon}}$, Algorithm~\ref{algo:verify} is asymptotically as fast as the modular
multiplication.
The dominant factor in both complexity is $\gO(n\log q \log\log q)$.\medskip

We shall mention that the verification cost in $(i)$ of Remark~\ref{rem:verifCostBetter} assumes the use of fast multiplication of
polynomials in order to check fast multiplication of polynomial of degree $\gO(n)$.  Even though this dependency is not a
problem in theory it might not be satisfactory in practice.  One solution would be to use a naive polynomial multiplication for
the extension field arithmetic but this further tightens the superiority of the verification.

\begin{rem}\label{rem:extension+naive}
  Assuming that extension field arithmetic is done naively using quadratic polynomial multiplication, Algorithm~\ref{algo:verify} remains
  faster than modular multiplication only when $ n^{1/3}< q $ in the dense case.
\end{rem}

We now propose an novel method that enables us to improve all the dense cases where an extension field is necessary, while not relying on
any polynomial arithmetic. More precisely, we show that fast verification does exist when $q<\log n$, which is of great interest
for the field $\cf_2$.  It is based on the evaluation of polynomials on matrices rather than scalars, combined with Freivalds
algorithm for verifying matrix multiplication \cite{Freivalds1979}.

Indeed, choosing $\alpha$ from an extension field inherently leads to depend on polynomial multiplication.  Instead of picking a
random point that is probably not a root of $\Delta=H-(FG)\bmod P$ when $\Delta\neq 0$, we pick a polynomial $R\in \cfq[X]$ of
degree $k<n$ that is probably not a divisor of $\Delta\neq 0$. To test whether $R$ divides $\Delta$, we evaluate $\Delta$ on the
companion matrix $C_R$ of $R$, defined by
\[C_R = \begin{pmatrix}
0 & 0 & \dotsb & 0 & -r_0\\
1 & 0 & \dotsb & 0 & -r_1\\
0 & 1 & \dotsb & 0 & -r_2\\
\vdots&\vdots&\ddots&\vdots&\vdots\\
0 & 0 & \dotsb & 1 & -r_k\\
\end{pmatrix}
\]
where $R = \sum_{i=0}^{k} r_i X^i$.  This strategy relies on the fact that $R$ is the minimal polynomial of its companion
matrix. Therefore, $R(C_R) = 0$ and any polynomial $\Delta$ such that $\Delta(C_R) = 0$ must be a multiple of $R$. In other words,
$R$ divides $\Delta$ if and only if $\Delta(C_R) = 0$. We will show that taking $R$ irreducible over $\cfq$ of degree
$k=\gO(\log n)$ makes this approach faster then the one using extension field when $\epsilon$ is constant. Furthermore, it will
extend the possibility to have fast verification for any fields, whatever the size of the polynomials.\medskip

To check whether $\Delta(C_R) = 0$, we need to evaluate $H$ and $(FG)\bmod P$ on $C_R$, and to verify that the evaluations match.
Of course, one cannot directly evaluate those polynomials on $C_R$ as it would cost $\gO(nk^{\omega})$ operations in $\cfq$ for
the dense case, where $\omega<2.3729$ is the best exponent for matrix multiplication \cite{LeGall14}. Since $k=\gO(\log n)$, this
would not give any improvement to Remark~\ref{rem:verifCostBetter}.

Instead, we rely on the so-called Freivald's technique to verify matrix multiplication \cite{Freivalds1979}. The idea is that the
matrix product $C=A\times B \in \gr^{k\times k}$ can be verified by asserting that $uC=(uA)\times B$ for a random vector
$u\in \{0,1\}^n$ with a probability of error of $1/2$. To assert that two polynomials evaluations on the matrix $C_R$ match, it is
sufficient to verify that their projection by the vector $u$ are equal. Given a degree-$n$ polynomials $H\in\cfq[X]$, one can
compute $uH(C_R)$ in $\gO(nk)$ operation in $\cfq$  using Horner evaluation: 
\begin{equation}\label{eq:HornerMatrix}
  uH(C_R) = u \sum_{i=0}^n h_i {C_R}^i =  \left(\sum_{i=1}^n h_i u{C_R}^{i-1}\right)C_R + u h_0.
\end{equation}
Since matrix-vector product with $C_R$ only costs $\gO(k)$ operations in $\cfq$, and Horner procedure only uses $n$ of those
matrix-vector products, the cost is clear.

\begin{rem}\label{rem:evalMatrix}
  It is sufficient to replace the evaluation of $F(\alpha)$ and $P(\alpha)$ by $uF(C_R)$ and $uP(C_R)$ in Algorithm~\ref{algo:eval} (\eval) to reach
	a complexity of $\gO(n(\#P+\deg(R)))$ operations in $\cfq$ for computing $u(FG \bmod P)(C_R)$ in the dense case.  More informally, it
  is sufficient to say that any of the operations in $\grext$ have now the cost of one matrix-vector product with $C_R$.
\end{rem}

\begin{thm}\label{thm:verifsmallfield}
Let $F$, $G$, $H$ and $P\in\cfq[X]$, as in Algorithm~\ref{algo:verify}. 
We can check whether $H = FG\bmod P$ in $\gO(\#Pn+n\log_q n\log\frac{1}{\epsilon})$ operations in $\cfq$ with a probability of
error at most $\epsilon$ if $H\neq FG$.
\end{thm}

\begin{proof}
  Let $0<\epsilon_1<\frac{1}{4}$ be a fixed probability. The algorithm needs two steps. First it computes with probability at
  least $1-\frac{1}{\epsilon_1}$ an irreducible polynomial $R$ of degree $d=\lceil\log_q\frac{2n}{\epsilon_1}\rceil$ using
  Proposition~\ref{prop:irrpoly}. Second, it computes $uH(C_R)$ and $u(FG\bmod P)(C_R)$ for some random vector $u\in\{0,1\}^d$. If
  both evaluations are distinct, the algorithm returns False. Otherwise, it repeats $\gO(\log\frac{1}{\epsilon})$ these two steps
  until one of the repetition fails. If this is the case it return false, otherwise the algorithm return true.

  If $H = (FG)\bmod P$, the algorithm always returns True. Let us assume that $H\neq (FG)\bmod P$, and let
  $\Delta = H-(FG)\bmod P$. For the algorithm to return True, each repetition must ensure that $uH(C_R) = u(FG\bmod P)(C_R)$. This
  may happen if either $R$ divides $\Delta$, whence $H(C_R) = (FG\bmod P)(C_R)$, or $R$ does not divide $\Delta$ but
  $uH(C_R) = u(FG\bmod P)(C_R)$. Since there are at least $q^d/2d$ irreducible polynomials of degree $d$ in $\cfq$ by
  Proposition~\ref{prop:nbirrpoly} and at most $n/d$ of them divide $\Delta$, the probability that $R$ divides $\Delta$ is at most
  $2n/q^d\le\epsilon_1$ provided $R$ is irreducible. Taking into account the probability that $R$ is not irreducible, the
  probability that $R$ divides $\Delta$ is at most $2\epsilon_1$. Then, using Freivalds standard argument, if $R$ does not divide
  $\Delta$, the probability $u\Delta(C_R) = 0$ is at most $\frac{1}{2}$. Altogether, the probability that one iteration returns
  True is at most $\frac{1}{2}+2\epsilon_1 < 1$. Therefore, the probability that
  $\log\frac{1}{\epsilon} / \log(\frac{1}{2}+2\epsilon_1)$ independent iterations all return True is at most $\epsilon$.

  Let us now analyse the complexity of the algorithm.  Since $\epsilon_1$ is a constant, the second step uses
  $\gO(\#Pn+n\log_q n)$ operations in $\cfq$ using Remark~\ref{rem:evalMatrix}.  The first step is negligible, even if naive
  polynomial arithmetic is used. Note that in this complexity, $\#Pn$ is the cost of Algorithm~\ref{algo:coef} (\coef). Since it is deterministic
  and only depends on $P$ and $F$, it can be called only once rather than at each iteration.

  We then get the complexity $\gO(\#Pn+n\log_q n\log\frac{1}{\epsilon})$.
\end{proof}

\begin{rem}
  Note that compared to using evaluation at $\alpha$ in an extension field, no operation depends on $\epsilon$.  If $\epsilon$ is
  fixed, the new method replaces a factor $\mul(\log_q n)$ in the complexity by $\log_q n$. Moreover our new approach requires
  only simple computations: additions of vectors, multiplication of a vector by a scalar and matrix-vector product with a
  companion matrix.  Furthermore, when $\#P$ and $\epsilon$ are constants, the verification is always faster than the
  modular multiplication, whatever the size of $q$.
\end{rem}

This new method still requires some polynomial arithmetic even if only naive polynomial multiplication is used.  This is because
the algorithm called to provide the degree-$d$ polynomial $R$ relies on polynomial products and GCDs to ensure that $R$ is
probably irreducible.  In order to remove the dependency to polynomial arithmetic, we can just choose a random monic degree-$d$
polynomial $R$ and compute the evaluation on $C_R$ even if $R$ is not irreducible. This implies to take several random polynomials $R$ to reach the target probability $\epsilon$.

\begin{cor}\label{cor:verifwithoutproduct}
  Let $F$, $G$, $H$ and $P\in\cfq[X]$, as in Algorithm~\ref{algo:verify}.  Without using any polynomial multiplication we can
  check whether $H = FG\bmod P$ in $\gO(\#Pn+n(\log_q n)^2\log\frac{1}{\epsilon})$ operations in $\cfq$ with a probability of
  error at most $\epsilon$ if $H\neq FG$.
\end{cor}

\begin{proof}
  In the proof of Theorem~\ref{thm:verifsmallfield}, we replace one evaluation on $C_R$ with $R$ irreducible with probability at
  least $1-\epsilon_1$ by few evaluations on several $C_R$ with $R$ random monic polynomial of degree $d$.  As the probability of
  a random monic polynomial $R$ to be irreducible is at least $\frac{1}{2d}$ by Proposition~\ref{prop:nbirrpoly}, we need to
  generate $\gO(d\log\frac{1}{\epsilon_1})$ random polynomials $R$ to reach a probability at least $1-\epsilon_1$ that at least
  one of them is irreducible.  Thus the evaluation part of the algorithm is repeated $\gO(d)=\gO(\log_q n)$ times since
  $\epsilon_1$ is constant.
\end{proof}

Even if this new approach does not improve the complexity from evaluation at $\alpha$ in an extension field using naive polynomial
multiplication, it allows to remove completely the dependency to any polynomial multiplication algorithm While this result might
not being seen useful as first sight, it will be used in Section~\ref{ssec:kaminski} to provide efficient verification for
polynomial multiplication. Indeed, in that case we will need to use verification with $P$ being of degree smaller than the input
degree.\medskip

This new method using companion matrix also works in the sparse case.  Indeed, any power $\alpha^t$ in Algorithm~\ref{algo:seval}
(\seval) are now replaced with $C_R^t$. However, only few powers $C_R^t$ with $1<t<n$ are relevant and we cannot compute all of
them as in the dense case.  This implies that computing $u(FG\bmod P)(C_R)$ instead of $(FG\bmod P)(C_R)$ is useless in that
case. Indeed, using fast exponentiation together with the structure of the powers of companion matrices~\cite{Gries1980} already
yield a complexity of $\gO((\#F\#P^{\lceil\frac{1}{\gamma}-1\rceil}+\#G)\log n\log_q^2 n)$ operations in $\cfq$, and we cannot
hope to lower this down by some random vector projection. In that case, using Freivald technique is useless and we have a better
probability of success. Choosing $\gO(\log\frac{n}{\epsilon}\log\frac{1}{\epsilon})$ polynomials $R$ at random, at least one of them is irreducible and does not divide $\Delta = H-FG\bmod P$ with probability $1-\epsilon$.  
This leads to an algorithm which does not use any polynomial product in the sparse case too. Even though it is asymptotically not as fast as the verification in an extension field where naive polynomial
arithmetic is used, it is still quasi-linear. 

\begin{cor}
  Let $P\in\cfq[X]$ be monic of degree $n$ and $F$, $G$, $H\in\cfq[X]$ of degree less than $n$ and sparsity at most $T$, and
  $0<\epsilon<1$.  Using a direct evaluation on a companion matrix, we can check whether $H = FG\bmod P$ with a probability of error
  at most $\epsilon$ if $H\neq FG$ in $\gO(T\#P ^{\lceil\frac{1}{\gamma}-1\rceil}\log n(\log_q \frac{n}{\epsilon})^3\log\frac{1}{\epsilon})$ operations
  in $\cfq$, without performing any polynomial product. 
\end{cor}

\section{Polynomial product verification}\label{sec:product}

In this section we study the simpler problem of verifying a classical polynomial multiplication.  Given three 
polynomials $F$, $G$ and $H\in\gr[X]$ of respective degrees $n$, $n$ and $2n$, the classical idea to verify $H = FG$ simply falls down to testing
$H(\alpha) = F(\alpha)G(\alpha)$ for some random $\alpha$ in a large enough set $\mathcal{S}$. As mentioned in the introduction,
this strategy may or may not have an optimal bit complexity, depending on the context. Here we are concerned with two difficulties
that arise in either the dense or the sparse cases.

If the polynomials are dense, the verification through evaluation requires a number of operations in $\mathcal{S}$ that is linear
in the input polynomials degree $n$. When $\gr$ has more than $n$ elements, taking $\mathcal{S}\subset\gr$ is sufficient to use
evaluation. However, multiplication in $\gr$ has not a linear bit complexity and best known results remain quasi-linear
\cite{cantor1991, harvey:polmul-nlogn,harvey2019}. The evaluation therefore leads to a quasi-linear bit complexity of
$\gO(n\imul(\log n)) = \gO(n\log n \log\log n)$.  When $\gr$ is too small, for instance with a small finite field, $\mathcal{S}$
is classically taken as a field extension of $\gr$, large enough to make it unlikely that $\alpha$ is a root of $H-FG$.
Therefore, each operation in $\mathcal{S}$ corresponds to an operation over $\gr[X]$ with non-negligible degree, meaning that the
number of operations in $\gr$ is no more linear in the inputs degree $n$.  As mentioned in the introduction, Kaminski's approach
\cite{Kaminski89} circumvents the later problem by replacing the evaluation with a computation in $\gr[X]/(X^i-1)$ for random
integer $i$ in a prescribed range. There, his algorithm is able to verify dense polynomial products with a linear number of
operations in $\gr$ whatever the ring size. However, the same difficulty as for large rings may arise. Since operations in $\gr$ do not have an linear bit complexity, unless say
$\gr=\cf_2$, this is not always sufficient to reach an optimal bit complexity for the verification.  In Section~\ref{ssec:kaminski}, we present
Kaminski's approach \cite{Kaminski89} and we provide a thorough analysis in the bit complexity model. In particular, we show that it
is possible to get optimal verification in the bit complexity model for any polynomial in $\gz[X]$ and for some polynomials in
$\cfq[X]$, depending on the relation between $q$ and $n$.

If the polynomials $F$, $G$ and $H$ are sparse with at most $T$ nonzero coefficients, the evaluation requires a number of
operations in $\mathcal{S}$ that is $\gO(T \log n)$. However the input bit size is given by the size of the exponents \emph{plus}
the size of the coefficients, that is $\gO(T+\log n)$ bits. Since $\mathcal S$ has to be of size at least $\gO(\log n)$, the bit
complexity of evaluation would be $\gO(T\log^2n)$ which is not even quasi-linear.  In Section~\ref{ssec:sparsprod} we develop a
novel method, already appearing in \cite{giorgi2020:sparsemul}, to verify sparse polynomial multiplication with a quasi-linear bit
complexity of $\gOt(T+\log n)$.

\subsection{Dense polynomial product verification}\label{ssec:kaminski}

In~\cite{Kaminski89}, Kaminski describes an algorithm to verify a polynomial product $H = FG\in\gr[X]$ using a linear number of
operations in $\gr$, regardless of its size.  His method chooses at random a polynomial $P$ that probably do not divides
$\Delta=H-FG$ if $\Delta\neq 0$. Then he verifies $H = FG\in\gr[X]/P$ using fast polynomial multiplication. Surprisingly, taking
$P$ of degree $o(n)$ in his algorithm enables to reach a linear number of operations in $\gr$. In the following we will often use
$\delta > 1.78107$ to be some constant value related to Euler's constant.

\subsubsection{Kaminski's algorithm}

The first step is to randomly select a polynomial from a fixed set, such that it most probably does not divide $\Delta=H-FG$ if
$\Delta\neq 0$. A standard approach could be to consider irreducible polynomials. This would be the direct generalization of the
evaluation method. However, Kaminski considers polynomials that are instead of the form $X^i-1$, for some integers $i>0$.  These
polynomials have two advantages: Reduction modulo $X^i-1$ has a linear cost and all their divisors are cyclotomic polynomials so
their least common multiple (lcm) has specific properties. 

\begin{prop}[\cite{Kaminski89}]
  For any integer set $I\subset\mathbb{N}$, $\prod_{i\in I}\Phi_i$ divides $\lcm \{X^i-1 : i\in I\}$, where $\Phi_i$ is the $i$-th
  cyclotomic polynomial in $\gr[X]$.
\end{prop}

Kaminski also gives a lower bound on the degree of $\lcm\{X^i-1 : i\in I\}$, depending in $I$. 
In particular the proposition implies that a nonzero polynomial, divisible by $k$ polynomials, of the form $X^i-1$, cannot have a
too small degree.
In the converse direction, a nonzero polynomial $\Delta$ of degree at most $2n$ cannot have too many divisors of the form $X^i-1$.
This is the content of Kaminski's main theorem. 

\begin{thm}[\cite{Kaminski89}]\label{kaminski1}
  Let $\Delta$ be a nonzero polynomial in $\gr[X]$ of degree $\leq 2n$ and $0<e<\frac{1}{2}$.  Let
  $k=\lceil 2\delta n^e\ln\ln (n^{1-e})\rceil$.  At most $k-1$ polynomials in the set $\{X^i-1|n^{1-e}\leq i<2n^{1-e}\}$ divide 
  $\Delta$.
\end{thm}

Kaminski's approach is then to choose a random integer $i\in[n^{1-e},2n^{1-e}[$, to reduce the input polynomials modulo
$X^i-1$ and to assert the equality in $\gr[X]/(X^i-1)$. We provide in Algorithm~\kaminski a more precise description of this 
approach.

    \begin{algorithm}
\caption{\kaminski}
\label{algo:kaminski}
\begin{algorithmic}[1]
  \Input $F,G,H \in \gr[X]$ of degree $n,n$ and $2n$; and $0<e<\frac{1}{2}$. 
 
  \Output True if $H = FG$, False with probability at least $1- (\lceil 2\delta n^e\ln\ln (n^{1-e})\rceil-1)/n^{1-e}$ otherwise. 

\State $i\gets$ random integer in $[n^{1-e}, 2n^{1-e}[$ 
\State $F_i,G_i,H_i\gets F\bmod X^i-1,G\bmod X^i-1,H\bmod X^i-1$ \label{kaminski:reduction} 
\State $M\gets F_iG_i$ \Comment Using a fast multiplication algorithm \label{kaminski:mul}
\State $M_i\gets M\bmod X^i-1$
\State \Return $M_i=H_i$ 
\end{algorithmic}
\end{algorithm}

\begin{thm}[\cite{Kaminski89}]\label{kaminski2}
  Let $F$, $G$ and $H\in\gr[X]$ of degree at most $n$, $n$ and $2n$, $0<e<\frac{1}{2}$ and an integer $k$ as in 
  Theorem~\ref{kaminski1}. Algorithm~\ref{algo:kaminski} uses $\gO(n)$ operations in $\gr$, and its failure probability is at most
  $(k-1)/n^{1-e}$ if $H\neq FG$. 
\end{thm}

\begin{rem}\label{rem:kaminskidetail}  
  To be more precise, Algorithm~\ref{algo:kaminski} requires $\gO(n)$ additions in $\gr$ at Step~\ref{kaminski:reduction} to compute the
  first three reductions modulo $X^i-1$, $\mul(n^{1-e})$ operations in $\gr$ to compute the product at Step~\ref{kaminski:mul},
  and $\gO(n^{1-e})$ additions in $\gr$ to compute the last reduction.
\end{rem}

One shall remark that the product in Step~3 must be computed with a subquadratic algorithm such that $\mul(n^{1-e})=\gO(n)$ since
$e < 1/2$.  If the parameter $e$ is taken close enough to $1/2$, Karatsuba's algorithm suffices to reach a linear number of
operations. 
The failure probability is $\gO(\log\log n/n^{1-2e})$, whence the need to have $e < 1/2$. We can bound this probability by $\gO(\frac{1}{n^{e'}})$ for any positive integer $e' < 1-2e$. In order to reach a probability $\epsilon$ of error, the algorithm should be repeated $\gO(\log_n\frac{1}{\epsilon})$ times. Note that this number of rounds is constant if $\epsilon$ is taken as $1/n^{\gO(1)}$.

The drawback of such approach is to crucially rely on a somewhat fast multiplication algorithm, and to perform multiplications of polynomials 
of degrees more than $ \sqrt n$. This means that optimal verification of the product of two degree-$n$ polynomials uses a product   
of polynomials of degrees close to $n$. In some contexts, such as verifying an implementation, relying on the same problem is
definitively problematic.\medskip

We note that all steps starting from Step~\ref{kaminski:mul} aim to verify $H_i = F_iG_i \bmod X^i-1$ deterministically.  It is 
easy to see that those steps can be replaced by our probabilistic modular product verification developed in 
Section~\ref{sec:verif}.  For polynomials over the integers or finite fields, this method does not require any polynomial
multiplications at all.

\begin{cor}\label{cor:kaminski+verif}
  If $\gr = \gz$ or a finite field, and $F$, $G$ and $H\in\gr[X]$ of degrees $n$, $n$ and $2n$. We can check whether $H = FG$ with
  a probability of failure at most $\epsilon$ if $H\neq FG$. This requires $\gO(n\log_n\frac{1}{\epsilon})$ additions in $\gr$
  plus $o(n\log_n\frac{1}{\epsilon})$ operations in $\gr$, without reverting to any polynomial multiplication.  In particular, the
  algorithm uses an optimal number of operations in $\gr$ when $\epsilon = 1/n^{\gO(1)}$.
\end{cor}

\begin{proof}
  We replace the last three steps of Algorithm~\ref{algo:kaminski} by a modular product verification, with a probability of failure at most $1/n$. Over $\gz$ or large finite fields, the complexity of this part is given by the dense version of Theorem~\ref{thm:verifmod} with $\#P=2$ and degree $i=\gO(n^{1-e})$. Over small finite fields, we rely instead on Corollary~\ref{cor:verifwithoutproduct}. In both cases, one can achieve a failure probability at most $1/n$ with at most $\gO(\log n)$ repetitions of the algorithm, for a total number of operations in $\gr$ that remains $o(n)$. 

  The total probability of failure of the modified algorithm is then  $1/n+\gO(1/n^{e'}) = \gO(1/n^{e'})$ for some $e'>0$. We can repeat this modified algorithm for $\gO(\log_n\frac{1}{\epsilon})$ rounds to get the announced failure probability and complexity. 
\end{proof}

\subsubsection{Analysis in the bit complexity model}

In \cite{Kaminski89}, Kaminski only details the algebraic complexity of its polynomial product verification, and no further
insights on the bit complexity are given. We now perform this analysis for polynomials over finite fields and over $\gz$. We
surprisingly prove that his algorithm remains linear in number of bit operations in many cases. For polynomials over $\cfq$, the algorithm fails to be linear only when $q$ is doubly exponentially larger than the degree. For polynomials over $\gz$, a similar condition applies. However, we are able to describe a variant of the algorithm that has linear bit complexity for polynomials with large coefficients. Hence we prove that polynomial product verification over $\gz$ has linear bit complexity in all cases.
Our variant is based on integer product verification, for which
Kaminski actually gives also in \cite{Kaminski89} a linear-time algorithm. Of course all those algorithms are therefore optimal.\medskip

The next theorem provides the bit complexity analysis of Kaminski's algorithm over finite fields.

\begin{thm}\label{thm:kaminskifields}
  Let $F$, $G$ and $H\in \cfq[X]$ of degrees $n$, $n$ and $2n$, and $0 < e < \frac12$. Algorithm~\ref{algo:kaminski} requires
  $\gO(n\log q +n^{1-e}\log q\log\log q)$ bit operations.  When $\log\log q = \gO(n^e)$, one can verify if $H = FG$ with failure
  probability at most $\epsilon$ if $H\neq FG$, using $\gO(n\log q\log_n\frac{1}{\epsilon})$ bit operations which is optimal when
  $\epsilon=1/n^{\gO(1)}$.
\end{thm} 
 
\begin{proof}
	We apply the count of operations given in Remark~\ref{rem:kaminskidetail}.
	The additions give the term $\gO(n\log q)$.
	The bit complexity of the product of degree-$\gO(n^{1-e})$ polynomials over $\cfq$ is $\mul_q(n^{1-e})=\gO(n^{1-e}\log q\log(n\log q))$, which is $\gO(n\log q+n^{1-e}\log q\log\log q)$. We obtain the claimed complexity.

        The second part directly follows from the observation that $\gO(\log_n\frac{1}{\epsilon})$ rounds of the algorithm yield a failure probability at most $\epsilon$.
\end{proof}

Note that the bound $\log\log q = \gO(n^e)$ to get a linear number of bit operations in $n \log q$ is only valid when using
the fastest known multiplication algorithm.
If we replace by a slower algorithm, the bound becomes smaller.  For instance, using Karatsuba's algorithm the product of
degree-$\gO(n^{1-e})$ polynomials uses $\gO(n^{(1-e)\log 3}\log q\log\log q)$ ring operations. For the algorithm to still have an
optimal complexity, we need that $n^{(1-e)\log 3}\log\log q = \gO(n)$. This implies $e\ge 1-1/\log 3\simeq 0.367$, and the bound
becomes $\log\log q = \gO(n^{1-(1-e)\log 3})$. If we take $e$ close to $1/2$, say $0.45$, the bound reads
$\log\log q = \gO(n^{0.13})$ while it is $\log\log q=\gO(n^{0.45})$ using the fastest multiplication algorithm.

Further, as mentioned previously, using a fast multiplication algorithm for the verification of a polynomial product is
problematic. We now analyse the bit complexity of our variant that does not use any polynomial product, that is of
Corollary~\ref{cor:kaminski+verif}. We show that the same complexity and the same bound on $q$ can be obtained without any
polynomial product.

\begin{rem}
    Let $F$, $G$, $H\in\cfq[X]$ of degrees $n$, $n$ and $2n$, and $0<e<\frac12$. Algorithm~\ref{algo:kaminski} can be implemented using a modular product verification and without any polynomial product. This variant has bit complexity
  $\gO(n\log q+n^{1-e}\log q\log\log q)$.  
  When $\log\log q = \gO(n^e)$, one can verify if $H = FG$ with failure probability at most $\epsilon$ if $H\neq FG$, using
  $\gO(n\log q\log_n\frac{1}{\epsilon})$ bit operations, which is optimal when $\epsilon=1/n^{\gO(1)}$, and without reverting to
  any polynomial product.
\end{rem}

\begin{proof} 
The proof simply consists in using Corollary~\ref{cor:verifwithoutproduct} in place of Remark~\ref{rem:kaminskidetail} in the previous proof.
\end{proof} 

\medskip

Now we consider $F$, $G$ and $H\in\gz[X]$ with $\norm{F},\norm{G},\norm{H}\leq C$. We first analyse the bit complexity of
Algorithm~\ref{algo:kaminski} and provide conditions for the algorithm to use a linear number of bit operations. Later we propose
a variant to be able to verify $H = FG$ with a linear number of bit operations for any integer polynomials..

\begin{thm}
  Let $F$, $G$ and $H\in\gz[X]$ of degrees $n$, $n$ and $2n$, and norms at most $C$, and
  $0<e<\frac12$. Algorithm~\ref{algo:kaminski} requires $\gO(n\log C+n^{1-e}\log C\log\log C)$ bit operations.  When
  $\log\log C = \gO(n^e)$, one can verify if $H = FG$ with failure probability at most $\epsilon$ if $H\neq FG$, using
  $\gO(n\log C\log_n\frac{1}{\epsilon})$ bit operations which is optimal when $\epsilon=1/n^{\gO(1)}$.
\end{thm}

\begin{proof}
  The first three reductions 
  require $\gO(n)$ additions in $\gz$ to compute $F_i$, $G_i$ and $H_i$, whose norms are at most $n^eC$. A careful computation of these additions
  using a binary tree uses $\gO(\sum_{i=1}^{\log n}\frac{n}{2^i}\log (iC)=\gO(n\log C)$ bit operations.  Then the polynomial
  product is performed with inputs of degree $n^{1-e}$ and norm $n^eC$.  As discussed in the introduction, it requires
  $\imul(n^{1-e}(\log (n^eC)+\log n^{1-e}))$ bit operations, that is
  $\gO(n^{1-e}(\log n+\log C)(\log n+\log\log C))=\gO(n\log C+n^{1-e}\log C\log\log C)$.  Finally the last reduction is performed
  with degree $2n^{1-e}$ and norm $n(n^eC)^2$ in $\gO(n\log C)$ bit operations.

  Repeating $\gO(\log_n\frac{1}{\epsilon})$ times the algorithm provides the second part of the theorem.
\end{proof}

As for polynomials over finite fields, the final computations can be replaced by a modular product verification. Here this yields a slightly better complexity. This improvement translates into an exponentially smaller
constraint on the norm $C$ for the algorithm to be optimal.

\begin{rem}
  Let $F$, $G$ and $H\in\gz[X]$, of degrees $n$, $n$ and $2n$ and norms at most $C$, and
  $0<e<\frac12$. Algorithm~\ref{algo:kaminski} can be implemented using a modular product verification and without any polynomial
  product. This variant has
  bit complexity $\gO(n\log C+n^{1-e}\log (C)\log\log\log (C))$. 
  When $\log\log\log C=\gO(n^e)$, one can verify if $H = FG$ with failure probability at most $\epsilon$ if $H = FG$, using
  $\gO(n\log C\log_n\frac{1}{\epsilon})$ bit operations, which is optimal when $\epsilon=1/n^{\gO(1)}$, and without reverting to
  any polynomial product.
\end{rem}

\begin{proof}
  The proof is once again similar, using the dense part of Theorem~\ref{thm:sparsemodverifZ} for the modular product
  verification. This verification is performed on polynomials of degrees $\gO(n^{1-e})$ and norm at most $n^eC$. Its bit
  complexity is then $\gO(n^{1-e}(\imul(\log(n\log C)+\log(C)\log\log(n\log C))))$ which is
  $\gO(n\log C+n^{1-e}\log(C)\log\log\log(C))$. This proves the first part of the remark. The second part relies on repetition of
  Algorithm~\ref{algo:kaminski}.
\end{proof}

As long as the coefficients are not insanely huge compared to the degree, the previous remark applies and the polynomial product
verification is linear. More precisely, this corresponds to $C$ ranging from $\gO(1)$ to $2^{2^{2^{\gO(n)}}}$.  To deal with this
extreme case of huge coefficients, we develop another approach that is valid as soon as $\log n = \gO(\log C)$.  This means that
all cases are covered with an optimal bit complexity. We shall mention that both methods are applicable when $C$ is ranging from
$n^{\gO(1)}$ to $2^{2^{2^{\gO(n)}}}$, which could be interesting when designing the most efficient implementation.

To treat the huge coefficient case, we rely on a result of Kaminski about the verification of the product of two integers.  His
technique is similar to the polynomial case: He reduces $s$-bit integers modulo $2^i-1$ for some $i$ between $s^{1-e}$ and
$2s^{1-e}$, and then performed the product with reduced integers.
 
\begin{thm}[\cite{Kaminski89}]\label{thm:kaminskiinteger}
  Let $a$, $b$, $c$ be integers of at most $s$, $s$ and $2s$ bits, $0<e<\frac{1}{2}$ and
  $k=\lceil 2\delta s^e\ln\ln (s^{1-e})\rceil$ where $\delta >1.78107$.  We can check whether $ab=c$ in $\gO(s)$ bit operations
  with a probability of error at most $(k-1)/s^{1-e}$ if $ab\neq c$.
\end{thm}

To verify a polynomial product $H = FG$ over $\gz$, we use the same idea as for computing the product. We use Kronecker
substitution. If we evaluate each polynomial on $\beta$ that is some large power of two, the coefficients of $FG$ can directly be
read on the digits of the integer $F(\beta)G(\beta)$. These evaluations at $\beta$ require no operation. The polynomial product verification is
thus reduced to an integer product verification $H(\beta) = F(\beta)G(\beta)$.

\begin{thm}
  Let $F$, $G$, $H\in\gz[X]$ of respective degrees $n$, $n$ and $2n$, and norm at most $C$. If $\log n = \gO(\log C)$, we can
  check whether $H = FG$ with failure probability at most $\epsilon$ if $H\neq FG$, using
  $\gO(n\log C\log_{n\log C}\frac{1}{\epsilon})$ bit operations, which is optimal when $\epsilon=1/n^{\gO(1)}$.
\end{thm}

\begin{proof}
  As $F$ and $G$ have norm $C$ and degree $n$, $FG$ has norm at most $nC^2$.  Let $\beta$ be the first power of $2$ greater than
  $nC^2$. Then $H = FG$ if and only if $H(\beta) = F(\beta)G(\beta)$.

  The integers $F(\beta)$, $G(\beta)$ and $H(\beta)$ have bit length $\gO(n\log\beta)=\gO(n\log (nC)) = \gO(n\log C)$ since
  $\log n = \gO(\log C)$. 
  As $\beta$ is a large enough power of $2$, the evaluation on $\beta$ does not require any operation.  Therefore all the cost
  comes from the verification of $F(\beta)G(\beta)=H(\beta)$. This is linear in the size of $F(\beta)$, $G(\beta)$ and $H(\beta)$
  by Theorem~\ref{thm:kaminskiinteger}, hence linear in $n\log C$.

  To get the appropriate probability bound, we use $\gO(\log_{n\log C}\frac{1}{\epsilon})$ round of this algorithm. This is
  supported by the fact that the probability bound in Theorem~\ref{thm:kaminskiinteger} is $1/s^{\gO(1)}$.
\end{proof}

\subsection{Quasi-linear sparse product verification}\label{ssec:sparsprod}

Given three sparse polynomials $F$, $G$ and $H$ in $\gr[X]$, we want to assert that $H = FG$. As already mentioned, evaluating the
polynomials at a random point $\alpha$ cannot yield a quasi-linear algorithm. Our approach is to take a random prime $p$ and to
verify the equality modulo $X^p-1$ through modular product verification.  This method is explicitly described in
Algorithm~\ref{algo:verif} that works over any large enough integral domain $\gr$.  We further extend the description and the
analysis of this algorithm for the specific cases $\gr=\gz$ and $\gr=\cfq$.

\begin{algorithm}
\caption{\verif}
\label{algo:verif}
\begin{algorithmic}[1]
  
  \Input $H,F,G\in\gr[X]$; $0<\epsilon<1$.
 
  \Output True if $H=FG$, False with probability at least $ 1-\epsilon$ otherwise.

\State Define $0<\epsilon_1<\frac{3}{10}$ and $0<\epsilon_2<1$ such that  $\frac{10\epsilon_1}{3} + (1-\frac{10\epsilon_1}{3})\epsilon_2\leq\epsilon$
\State $n\gets\deg(H)$
\If{$\#H>\#F\#G$ or $n\neq \deg(F)+\deg(G)$} \Return False\EndIf
	\State $\lambda\gets\max(21,\frac{1}{\epsilon_1}(\#F\#G+\#H)\ln n)$
\State $p\gets \rp(\lambda,\frac{5\epsilon_1}{3})$
\State $(F_p,G_p,H_p) \gets(F \bmod X^p-1,~G \bmod X^p-1,~H\bmod X^p-1) $
\State \Return True if $H_p = (F_pG_p)\bmod X^p-1$, False otherwise \Comment using Theorem \ref{thm:verifmod} with probability $\epsilon_2$
\end{algorithmic}
\end{algorithm}

\begin{thm}\label{thm:algogénéral}
  If $\gr$ is an integral domain of size $\geq \frac{2}{\epsilon_1\epsilon_2}(\#F\#G+\#H)\ln(n)$, Algorithm~\ref{algo:verif} works
  as specified. Assuming that $n = \deg(H)$ and $T = \max(\#F,\#G,\#H)$, it requires $\gO(T\log(\frac{1}{\epsilon}T\log n))$ operations in
  $\gr$, and $\gO(T\log n \log\log(\frac{1}{\epsilon}T\log n))$ bit operations plus
  $\gO(\log\frac{1}{\epsilon}\log^3(\frac{1}{\epsilon}T\log n)\log^2\log(\frac{1}{\epsilon}T\log n))$ bit operations to obtain a
  prime $p$.
\end{thm}

\begin{proof}
Step~3 dismisses two trivial mistakes and ensures that $n$ is a bound on the degree of each polynomial.

If $H = FG$, the algorithm always returns True.  Otherwise, there are two sources of failure.  Either $X^p-1$ divides $H-FG$.
Since this polynomial has at most $\#H + \#F\#G$ terms, this failure occurs with probability at most $\frac{10\epsilon_1}{3}$ by
Proposition~\ref{prop:choixdep}.  Or $X^p-1$ does not divide $H-FG$ but the modular product verification fails. This occurs with
probability at most $\epsilon_2$.  Altogether, the failure probability is at most
$\frac{10\epsilon_1}{3} + (1-\frac{10\epsilon_1}{3})\epsilon_2\leq\epsilon$.

To analyse the complexity, we consider $\epsilon_1,\epsilon_2\sim\epsilon$ (for example $\epsilon_1=\frac{3\epsilon}{20}$ and
$\epsilon_2=\frac{\epsilon}{2}$).  Let us remark that $p=\gO(\frac{1}{\epsilon}T^2\log n)$.  To get the prime $p$, Step~5 requires
only $\gO(\log\frac{1}{\epsilon}\log^3p\log^2\log p)$
bit operations by Proposition~\ref{prop:rdprime}. This gives the announced complexity once $\log
p$ is replaced by $\gO(\log(\frac{1}{\epsilon} T\log n))$. 

The operations in Step~6 are $T$ divisions by $p$ on integers bounded by $n$.  Their cost is $\gO(T\frac{\log n}{\log p}\imul(\log
p))=\gO(T\log n\log\log p)$ bit operations, that is $\gO(T \log n\log\log (\frac{1}{\epsilon}T\log n)))$, plus
$T$ additions in $\gr$.

In Step~7, $F_p$, $G_p$ and $H_p$ have degree $p=\gO(\frac{1}{\epsilon}T^2\log n)$ and at most $T$ monomials.  They are still
sparse and we can use the sparse version of Theorem~\ref{thm:verifmod} with $P=X^p-1$.  The verification of
$H_p = F_pG_p\bmod X^p-1$ thus requires $\gO(T\log p)=\gO(T\log(\frac{1}{\epsilon}T\log n))$ operations in $\gr$.  Other steps
have negligible cost.
\end{proof}

To clarify the complexity, we will use the notation $\gOeps(f(n))$ as a shortcut for $\gO(f(n)\log^k\frac{1}{\epsilon})$ for some
$k$. 
Using this notation, the complexity of Algorithm~\ref{algo:verif} becomes $\gOeps(T\log(T\log n))$ operations in $\gr$ plus
$\gOeps(T\log n\log\log (T\log n)))$ bit operations as getting the prime $p$ is logarithmic in $T$ and $\log n$.\medskip

The rest of the section is dedicated to the bit complexity analysis of this algorithm over integers or finite fields.  Our goal is
to have bit complexities that are as close as possible to linear. To ease the comparison with truly linear complexity, we express
these bit complexities in terms of the total bit size $s$ of the input. A degree-$n$ polynomial with $T$ monomials has bit size
$s = \gO(T(\log n + \log q))$ if it has coefficients in $\cfq$, and $s = \gO(T(\log n + \log C))$ if it has coefficients in $\gz$
of absolute value at most $C$.

We first note that reducing the input polynomials modulo $X^p-1$ at Step~6 is already non-linear. Indeed, we proved that this step has bit complexity $\gOeps(T\log n\log\log(T\log n))$, which is $\gOeps(s\log\log s)$. We shall prove that in some cases, this step is actually the dominant term in the complexity.

We begin with the analysis over the integers.

\begin{cor} \label{cor:algosurz} Let $F$, $G$ and $H\in\gz[X]$ of degree at most $n$, with norm at most $C$ and sparsity at
  most $T$.
  Then Algorithm~\ref{algo:verif} 
	
        has bit complexity
	$\gOeps(s\log s\log\log s)$, where
  $s = T(\log n+\log C)$ is the input size. 
\end{cor}

\begin{proof}
        The modification only concerns Step~7, where we use Theorem~\ref{thm:sparsemodverifZ} for the modular product verification with

        $P=X^p-1$ 
        and $F_p$, $G_p$, $H_p$ that have sparsity $T$ and norm $TC$.
	So this step costs 
        \[\gOeps(T\log p\imul(\log(p\log C)) + T\log(TC)\log\log(p\log TC)).\]
        Since $T\le n$, $T\log p = \gOeps(T\log n) = \gOeps(s)$. And $\log(p\log C) = \gOeps(\log(T\log n\log C)) = \gOeps(\log(T\log n) + \log\log C) = \gOeps(\log s)$. Thus the first term is $\gOeps(s\log s\log\log s)$.
        Also, $T\log(TC) = \gO(T\log nC) = \gO(s)$.

        As $\log(p\log TC) = \gOeps(\log(T\log n)+\log\log C) = \gOeps(\log S)$, the second term is 
        $\gOeps(s\log\log s)$.

        Since Step~6 is unchanged and has bit complexity $\gOeps(s\log\log s)$, the result follows.
\end{proof}

The complexity is actually better for very sparse polynomials.

\begin{rem}\label{rem:verifrèscreusedansz}
  If $F$, $G$, $H\in\gz[X]$ of bit size $s$ have sparsity at most $T = \Theta(\log^k n)$ for some $k$, Algorithm~\ref{algo:verif}
  has bit complexity $\gOeps(s\log\log s)$.
\end{rem}

\begin{proof}
The input size is $s=\Theta(\log^{k+1}n+\log^kn\log C)$. In this case, $\log p=\gOeps(\log\log n)$. In the previous proof, there is one dominant term of order $\gOeps(s\log s\log\log s)$, while the other terms are already of order $\gOeps(s\log\log s)$. It is sufficient to prove that with the new assumption, the dominant term is also $\gOeps(s\log\log s)$.

The dominant term $\gOeps(s\log s\log\log s)$ in the complexity comes from the term $\gOeps(T\log p\imul(\log(p\log C)))$. Since $\log(p\log C) = \gOeps(\log\log n+\log\log C)$, this dominant term becomes 
\[\gOeps(\log^k n\log\log n(\log\log n+\log\log C)\log(\log\log n+\log\log C)).\]
Note that $\log\log n$ and $\log\log C$ are both $\gO(\log s)$, therefore this can be rewritten $\gOeps(\log^k n\log^2s \log\log s)$. Since $\log^k n = \gO(s^{k/(k+1)})$, this yields $\gOeps(s\log\log s)$.
\end{proof}

We now switch to polynomials over finite fields. There are more cases to consider, depending on the size of the field with respect
to the degree and sparsity of the inputs.  The first easy case is the case of large finite fields: If there are enough points for
the evaluation, the generic algorithm keeps its guarantee of success while offering a quasi-linear bit complexity.

\begin{cor}\label{cor:grandcf} 
  Let $F$, $G$ and $H\in\cfq[X]$ of degree at most $n$ and sparsity at most $T$ where
  $q>\frac{2}{\epsilon_1\epsilon_2}(\#F\#G+\#H)\ln n$.  Then Algorithm~\ref{algo:verif} has bit complexity $\gOeps(s\log^2(s))$
  where $s = T(\log n+\log q)$ is the input size.
\end{cor}

\begin{proof}
    It is still enough to analyse Step~7. 
  Each ring operation in $\cfq$ costs
  $\gO(\log(q)\log\log(q))$ bit operations which implies that the bit complexity of Step~7 is
  $\gO_\epsilon(T\log(T\log n)\log(q)\log\log(q))$. Since both $T\log q$ and $T\log n$ are $\gO(s)$ and
  $\log\log q = \gO(\log s)$, the result follows.
\end{proof}

If the field is not large enough, we need to use some extension field. This slightly modifies the algorithm but actually yields a
better complexity bound than for large finite fields. This is due to the fact that in that case, we choose an extension of the
exact appropriate size. Note that the probability of success remains unchanged.

\begin{cor}
  Let $F$, $G$ and $H\in\cfq[X]$ of degree at most $n$ and sparsity at most $T$ where
  $q<\frac{2}{\epsilon_1\epsilon_2}(\#F\#G+\#H)\ln n$.
  Algorithm~\ref{algo:verif} 
   has bit complexity $\gOeps(s\log s\log\log s)$, where $s = T(\log n+\log q)$ is the input size.
\end{cor}

\begin{proof} \sloppy
  By Corollary~\ref{cor:smallfieldbyextension} Step~7 requires
  $\gOeps(T\log p\mul_q(\log_q p)+(\log_q p)^2\mul_q(\log_q p)(\log q+\log\log_q p))$ operations in $\cfq$. Since
  $\log p=\gOeps(\log(T\log n))=\gOeps(\log s)$ and $\log q=\gOeps(\log s)$ too, the second term is polylogarithmic in $s$.  As
  $\log p=\gOeps(\log (T\log n))$ the first term is $\gOeps(T\log (T\log n) \mul_q(\log_q(T\log n)))$. Since
  $\log(T\log n) = \gO(\log n)$, $T\log(T\log n) = \gO(s)$. Furthermore, $\log(T\log n) = \gO(\log s)$ and the first term
  simplifies to $\gOeps(s\mul_q(\log_q s))$. Now $\mul_q(\log_q s) = \gO(\log s\log\log s)$. Altogether
  $T\log p\mul_q(\log_q p) = \gOeps(s\log s\log\log s)$.
The result follows.
\end{proof}

Again, we note that for very sparse polynomials over some fields, the complexity is even better.

\begin{rem}
  Let $F$, $G$ and $H\in\cfq[X]$ of degree at most $n$ and sparsity at most $T$, where
  $q<\frac{2}{\epsilon_1\epsilon_2}(\#F\#G+\#H)\ln n$.  The bit complexity of Algorithm~\ref{algo:verif} is
\begin{itemize}
	\item[(i)] $\gOeps(s\log s)$ if $\log_q (T\log n)=\gO(1)$,
	\item[(ii)] $\gOeps(s\log\log s)$ if $T = \Theta(\log^k n)$ for some constant $k$.
\end{itemize}
\end{rem}

\begin{proof} \sloppy
  The most significant term in the complexity is $\gOeps(T\log (T\log n) \mul_q(\log_q(T\log n))$. In the first case, it becomes
  $\gOeps(T\log(T\log n)\mul_q(1)) = \gOeps(T\log(T\log n)\log q\log\log q)$. As $\log q = \gOeps(\log\log n)$,
  $T\log q\log\log q = \gOeps(s)$ and the complexity becomes $\gOeps(s\log s)$.
  In the second case, 
  the most significant term can be bounded by $\gOeps(T\log^3(T\log n))$. But $T = \gO(s^{k/(k+1)})$, and this most significant
  term becomes $\gOeps(s)$ only. The global bit complexity is then dominated by Step~6 and is $\gOeps(s\log\log s)$.
\end{proof}

To conclude, the bit complexity of Algorithm~\ref{algo:verif} over integers or finite fields range from $\gOeps(s\log\log s)$ in
the most favorable cases, to $\gOeps(s\log^2 s)$ in more complicated situations. We note that in the best cases, the complexity is
actually dominated by the cost of the modular reduction of the exponents of the input polynomials.

\begin{rem}
  Verification of a sparse product is always faster than computing the sparse product over $\gz$ or $\cfq$.
\end{rem}

\begin{proof}
  Assuming $s = T(\log n+\log \zeta)$ to be the input size of the sparse polynomial $F$,$G$ and $H$. Over $\gz$ we have
  $\log \zeta = \log C$ where C is the norm of the coefficients, while $\log \zeta= \log q$ when in $\cfq$. The best know result
  for computing the product $FG$ needs $\gOeps(s\log^2(s)\log^2(T)(\log T + \log\log s))$ bit operations
  \cite{giorgi2020:sparsemul}. Taking the worst case complexity for our verification yields a cost of $\gOeps(s\log^2 s)$. This
  means that we are always faster by a factor $\gO(\log^2(T)(\log T + \log\log s))$. Of course, for some small finite fields we
  are even beyond this value.
\end{proof}

\end{document}